\documentclass[12pt,draftcls,onecolumn]{IEEEtran}

\usepackage{enumerate}

\usepackage{amsbsy}
\usepackage{amssymb}
\usepackage{amscd}
\usepackage{amsmath}
\usepackage{graphics}
\usepackage{graphicx}
\usepackage{epstopdf}
\usepackage{psfrag}
\usepackage{amsthm}
\usepackage{verbatim}
\def\bPhi{{\mathbf{\Phi}}}

\def\bGam{{\mathbf{\Gamma}}}

\def\bSig{{\mathbf{\Sigma}}}
\def\bPsi{{\mathbf{\Psi}}}

\def\bc{{\mathbf{c}}}

\def\be{{\mathbf{e}}}
\def\bg{{\mathbf{g}}}
\def\bh{{\mathbf{h}}}

\def\bv{{\mathbf{v}}}
\def\bw{{\mathbf{w}}}
\def\bx{{\mathbf{x}}}
\def\by{{\mathbf{y}}}
\def\bz{{\mathbf{z}}}
\def\bA{{\mathbf{A}}}
\def\bB{{\mathbf{B}}}
\def\bC{{\mathbf{C}}}
\def\bD{{\mathbf{D}}}

\def\bF{{\mathbf{F}}}
\def\bG{{\mathbf{G}}}
\def\bH{{\mathbf{H}}}
\def\bI{{\mathbf{I}}}

\def\bK{{\mathbf{K}}}

\def\bN{{\mathbf{N}}}

\def\bP{{\mathbf{P}}}
\def\bQ{{\mathbf{Q}}}

\def\bS{{\mathbf{S}}}
\def\bT{{\mathbf{T}}}
\def\bU{{\mathbf{U}}}

\def\bW{{\mathbf{W}}}

\def\b0{{\boldsymbol{0}}}

\newlength{\figwidth}
\newlength{\figwidthb}
\theoremstyle{definition}
\newtheorem{thm}{Theorem}
\newtheorem{lem}{Lemma}
\newtheorem{rem}{Remark}

\newtheorem{cor}{Corollary}

\usepackage{setspace}

\ifCLASSINFOpdf
\else
\fi

\begin{document}
\DeclareGraphicsExtensions{.eps}

\title{On the Optimality of Linear Precoding for Secrecy in the MIMO Broadcast
Channel}

\author{\IEEEauthorblockN{S. Ali. A. Fakoorian and A.~Lee~Swindlehurst\thanks{This work was supported by the U.S. Army Research Office under the
Multi-University Research Initiative (MURI) grant W911NF-07-1-0318,
and by the U.S. National Science Foundation under grant CCF-1117983.}}
}

\maketitle

\begin{abstract}
\begin{singlespace}
We study the optimality of linear precoding for the two-receiver
multiple-input multiple-output (MIMO) Gaussian broadcast channel (BC)
with confidential messages.  Secret dirty-paper coding (S-DPC) is
optimal under an input covariance constraint, but there is no
computable secrecy capacity expression for the general MIMO case under
an average power constraint.  In principle, for this case, the secrecy
capacity region could be found through an exhaustive search over the
set of all possible matrix power constraints. Clearly, this search,
coupled with the complexity of dirty-paper encoding and decoding,
motivates the consideration of low complexity linear precoding as an
alternative.  We prove that for a two-user MIMO Gaussian BC under an
input covariance constraint, linear precoding is optimal and achieves the same
secrecy rate region as S-DPC if the input covariance constraint satisfies a
specific condition, and we characterize the corresponding optimal
linear precoders.  We then use this result to derive a  closed-form sub-optimal
algorithm based on linear precoding for an average power constraint.
Numerical results indicate that the secrecy rate region achieved
by this algorithm is close to that obtained by the optimal
S-DPC approach with a search over all suitable input covariance
matrices.
\end{singlespace}
\end{abstract}

\section{Introduction}
The work of Wyner \cite{Wyner} led to the development of the notion
of secrecy capacity, which quantifies the maximum rate at which a
transmitter can reliably send a secret message to a receiver,
without an eavesdropper being able to decode it.  More recently,
researchers have considered secrecy for the two-user broadcast
channel, where each receiver acts as an eavesdropper for the
independent message transmitted to the other.  This problem
was addressed in \cite{Yates08}, where inner and outer bounds for the
secrecy capacity region were established. Further work in \cite{Liu09}
studied the multiple-input single-output (MISO) Gaussian case, and
\cite{LiuLiu10} considered the general MIMO Gaussian case. It was
shown in \cite{LiuLiu10} that, under an input covariance constraint,
both confidential messages can be simultaneously communicated at their
respective maximum secrecy rates, where the achievablity is obtained
using secret dirty-paper coding (S-DPC). However, under an average
power constraint, a computable secrecy capacity expression for the
general MIMO case has not yet been derived. In principle, the secrecy
capacity for this case could be found by an exhaustive search
over the set of all input covariance matrices that satisfy the average
power constraint \cite{LiuLiu10}. Clearly, the complexity associated with such a
search and the implementation of dirty-paper encoding and decoding
make such an approach prohibitive except for very simple scenarios,
and motivates the study of simpler techniques based on linear
precoding.

While low-complexity linear transmission techniques have been
extensively investigated for the broadcast channel (BC) without
secrecy constraints, e.g., \cite{Prof04}-\cite{Wiesel}, there has been
relatively little work on considering secrecy in the design of linear
precoders for the BC case.  In \cite{oursGlobec11}, we considered
linear precoders for the MIMO Gaussian broadcast channel with
confidential messages based on the generalized singular value
decomposition (GSVD) \cite{KhistiMIMO,oursISIT12}. It was shown
numerically in \cite{oursGlobec11} that, with an optimal allocation
of power for the GSVD-based precoder, the achievable secrecy rate
is very close to the secrecy capacity region.

In this paper, we show that for a two-user MIMO Gaussian BC with
arbitrary numbers of antennas at each node and under an input
covariance constraint, linear precoding  is optimal and achieves the same secrecy rate
region as S-DPC for certain input covariance constraints, and we
derive an expression for the optimal precoders in these scenarios.  We
then use this result to develop a sub-optimal closed-form algorithm
for calculating linear precoders for the case of average power
constraints.  Our numerical results indicate that the secrecy rate
region achieved by this algorithm is close to that obtained
by the optimal S-DPC approach with a search over all suitable input
covariance matrices.

In Section \ref{secII}, we describe the model for the MIMO
Gaussian broadcast channel with confidential messages and the optimal
S-DPC scheme, proposed in \cite{LiuLiu10}. In Section \ref{secIII},
we consider a general MIMO broadcast channel under a matrix
covariance constraint, we derive the
conditions under which linear precoding is optimal and
achieves the
same secrecy rate region as S-DPC, and we find the corresponding
optimal precoders.  We then present our sub-optimal algorithm for
designing linear precoders for the case of an average power constraint
in Section~\ref{secIV}, followed by numerical examples in
Section~\ref{secV}. Section \ref{secVI} concludes the paper.

\textbf{Notation:} Vector-valued random variables are written with
non-boldface uppercase letters ({\em e.g.,} $X$), while the
corresponding non-boldface lowercase letter ($\bx$) denotes a specific
realization of the random variable. Scalar variables are written with
non-boldface (lowercase or uppercase) letters.  The Hermitian (i.e.,
conjugate) transpose is denoted by $(.)^H$, the matrix trace by Tr(.),
and \textbf{I} indicates an identity matrix.  The inequality $\bA
\succ \bB$ ($\bA\succeq\bB$) means that $\bA-\bB$ is Hermitian positive
(semi-)definite. Mutual information between the random variables $A$ and
$B$ is denoted by $I(A;B)$, $\mathbb{E}$ is the expectation operator,
and $\mathcal{CN}(0,\sigma^2)$ represents the complex circularly
symmetric Gaussian distribution with zero mean and variance
$\sigma^2$.

\section{Broadcast Channel and S-DPC}  \label{secII}
We consider a two-receiver multiple-antenna Gaussian broadcast channel
with confidential messages, where the transmitter, receiver~1 and
receiver~2 possess $n_t$, $m_1$, and $m_2$ antennas, respectively. The
transmitter has two independent confidential messages, $W_1$ and
$W_2$, where $W_1$ is intended for receiver~1 but needs to be kept
secret from receiver~2, and $W_2$ is intended for receiver~2 but needs
to be kept secret from receiver~1 \cite{LiuLiu10}.

The signals at each receiver can be written as:
\begin{align}\label{lin1}
\begin{split}
\by_1=\bH\bx+ \bz_1\\
\by_2=\bG\bx+ \bz_2
\end{split}
\end{align}
where $\bx$ is the $n_t\times1$ transmitted signal, and
$\bz_i\in\mathbb{C}^{m_i\times1}$ is white Gaussian noise
at receiver~$i$ with independent and identically distributed
entries drawn from $\mathcal{CN}(0, 1)$. The channel matrices
$\bH\in\mathbb{C}^{m_1\times n_t}$ and $\bG\in\mathbb{C}^{m_2\times
n_t}$ are assumed to be unrelated to each other, and known at all
three nodes.  The transmitted signal is subject to an average
power constraint when
\begin{equation}\label{lin2}
\mathrm{Tr}(\mathbb{E}\{XX^H\})=\mathrm{Tr}(\bQ)\leq P_t
\end{equation}
for some scalar $P_t$, or it is subject to a matrix power constraint
when \cite{LiuLiu10,Bustin}:
\begin{equation}\label{lin3}
\mathbb{E}\{XX^H\}=\bQ\preceq \bS
\end{equation}
where $\bQ$ is the transmit covariance matrix, and $\bS \succeq 0$.
Compared with the average power constraint, (\ref{lin3}) is rather precise
and inflexible, although for example it does allow for the incorporation
of per-antenna power constraints as a special case.

It was shown in \cite{Yates08} that for any jointly distributed $(V_1,
V_2,X)$ such that $(V_1, V_2)\rightarrow X \rightarrow(Y_1,Y_2)$ forms
a Markov chain and the power constraint over $X$ is satisfied, the
secrecy rate pair $(R_1,R_2)$ given by
\begin{align}\label{lin4}
\begin{split}
R_1&=I(V_1;Y_1)-I(V_1;V_2,Y_2)\\
R_2&=I(V_2;Y_2)-I(V_2;V_1,Y_1)
\end{split}
\end{align}
is achievable for the MIMO Gaussian broadcast channel given by
(\ref{lin1}), where the auxiliary variables $V_1$ and $V_2$ represent
the precoding signals for the confidential messages $W_1$ and $W_2$,
respectively \cite{LiuLiu10}. In \cite{Yates08}, the achievablity of
the rate pair~(\ref{lin4}) was proved.

Liu \emph{et al.} \cite{LiuLiu10} analyzed the above secret
communication problem under the matrix power-covariance constraint
(\ref{lin3}). They showed that the secrecy capacity region
$\mathcal{C}_s(\bH,\bG,\bS)$ is rectangular. This interesting result
implies that under the matrix power constraint, both confidential
messages $W_1$ and $W_2$ can be \emph{simultaneously} transmitted at
their respective maximal secrecy rates, as if over two separate MIMO
Gaussian wiretap channels. To prove this result, Liu \emph{et al.}
showed that the secrecy capacity of the MIMO Gaussian wiretap channel can also be achieved via a coding
scheme that uses artificial noise and random binning \cite[Theorem 2]{LiuLiu10}.

Under the matrix power constraint (\ref{lin3}), the achievablity of
the optimal corner point $(R_1^*,R_2^*)$ given by \cite[Theorem 1]{LiuLiu10}
\begin{align}\label{lin5}
\begin{split}
R_1^* & =\max_{0\preceq\bK_t\preceq\bS}\log\left|\bH\bK_t\bH^H+\bI\right|-\log\left|\bG\bK_t\bG^H+\bI\right|\\
R_2^* & =\log\left|\frac{\bG\bS\bG^H+\bI}{\bH\bS\bH^H+\bI}\right|+R_1^*
\end{split}
\end{align}
is obtained using dirty-paper coding based on double binning, or as
referred to in \cite{LiuLiu10}, secret dirty paper coding
(S-DPC). More precisely, let $\bK_t^*\succeq\b0$ maximize (\ref{lin5}),
and let
\begin{equation}\label{lin6}
V_1=U_1+\bF U_2 \qquad V_2=U_2 \qquad X=U_1+U_2 \; ,
\end{equation}
where $U_1$ and $U_2$ are two independent Gaussian vectors with zero
means and covariance matrices $\bK_t^*$ and $\bS-\bK_t^*$,
respectively, and the precoding matrix $\bF$ is defined as
$\bF=\bK_t^*\bH^H\left(\bH\bK_t^*\bH^H+\bI\right)^{-1}\bH$.  One can
easily confirm the achievablity of the corner point $(R_1^*,R_2^*)$ by
evaluating (\ref{lin4}) for the above random variables and noting that
in (\ref{lin1}), $X=U_1+U_2$.  Note that under the matrix power
constraint $\bS$, the input covariance matrix that achieves the corner
point in the secrecy capacity region satisfies $\bQ=\bS$
\cite{LiuLiu10}.

The matrix $\bK_t$ that maximizes (\ref{lin5}) is given by \cite{LiuLiu10,Bustin}
\begin{align}\label{lin8}
\bK^*_t= \bS^{\frac{1}{2}}\bC \left[
\begin{array}{ccc}
(\bC_1^H\bC_1)^{-1} & \b0\\
\b0 & \b0
\end{array}
\right]\bC^H \bS^{\frac{1}{2}}
\end{align}
where $\bC=[\bC_1 \; \bC_2]$ is an invertible\footnote{Note that $\bC$
is invertible since both components of the pencil (\ref{lin9}) are
positive definite.} generalized eigenvector matrix of the pencil
\begin{align}\label{lin9}
\left(\bS^{\frac{1}{2}}\bH^H\bH\bS^{\frac{1}{2}}+\bI\;,\; \bS^{\frac{1}{2}}\bG^H\bG\bS^{\frac{1}{2}}+\bI\right)
\end{align}
satisfying \cite{Horn}
\begin{align}\label{lin10}
\begin{split}
&\bC^H\left[\bS^{\frac{1}{2}}\bH^H\bH\bS^{\frac{1}{2}}+\bI\right]\bC=\mathbf{\Lambda}\\
&\bC^H\left[\bS^{\frac{1}{2}}\bG^H\bG\bS^{\frac{1}{2}}+\bI\right]\bC=\bI \; ,
\end{split}
\end{align}
where
$\mathbf{\Lambda}=\text{diag}\{\lambda_1,...,\lambda_{n_t}\}\succ\b0$
contains the generalized eigenvalues sorted without loss of generality
such that
$$\lambda_1\geq...\geq\lambda_b >1 \geq
\lambda_{b+1}\geq...\geq\lambda_{n_t}>0 \; .$$
The quantity $b$ denotes the number of generalized eigenvalues greater
than one $(0\leq b \leq n_t)$, and defines the following matrix partitions:
\begin{equation}
\mathbf{\Lambda}  = \left[
\begin{array}{ccc}
\mathbf{\Lambda}_1 & 0\\
0 & \mathbf{\Lambda}_2
\end{array}
\right] \qquad
\bC = [\bC_1\; \,\bC_2] \; ,
\label{lin11}
\end{equation}
where $\mathbf{\Lambda}_1=\text{diag}\{\lambda_1,...,\lambda_b\}$,
$\mathbf{\Lambda}_2=\text{diag}\{\lambda_{b+1},...,\lambda_{n_t}\}$,
$\bC_1$ contains the $b$ generalized eigenvectors corresponding to
$\mathbf{\Lambda}_1$ and $\bC_2$ the $(n_t-b)$ generalized
eigenvectors corresponding to $\mathbf{\Lambda}_2$.  Now, by applying
(\ref{lin8}) in (\ref{lin5}), the corner rate pair $(R_1^*,R_2^*)$ can
be calculated as (\cite[Theorem 3]{LiuLiu10})
\begin{align}\label{lin13}
\begin{split}
&R_1^*=\log\left|\mathbf{\Lambda}_1\right|\\
&R_2^*=-\log\left|\mathbf{\Lambda}_2\right|  \;.
\end{split}
\end{align}

For the average power constraint in~ (\ref{lin2}), there is no
computable secrecy capacity expression for the general MIMO case. In
principle the secrecy capacity region for the average power
constraint, $\mathcal{C}_s(\bH,\bG,P_t)$, could be found
through an exhaustive search over all suitable matrix power constraints
\cite{LiuLiu10},\cite[Lemma 1]{Weingarten}:
\begin{align}\label{lin14}
\mathcal{C}_s(\bH,\bG,P_t)=\bigcup_{\bS\succeq 0, \mathrm{Tr}(\bS)\leq P_t}
\mathcal{C}_s(\bH,\bG,\bS) \; .
\end{align}
For any given semidefinite $\bS$, $\mathcal{C}_s(\bH,\bG,\bS)$ can be
computed as given by (\ref{lin13}). Then, the secrecy capacity region
$\mathcal{C}_s(\bH,\bG,P_t)$ is the convex hull of all of the obtained
corner points using (\ref{lin13}).

The complexity associated with such a search, as well as that required
to implement dirty-paper encoding and decoding, are the main drawbacks
of using S-DPC to find the secrecy capacity region
$\mathcal{C}_s(\bH,\bG,P_t)$ for the average power constraint. This
makes linear precoding (beamforming) techniques an attractive
alternative because of their simplicity.  To address the performance
achievable with linear precoding, we first describe the conditions
under which linear precoding is optimal in attaining the same secrecy rate region that is achievable
via S-DPC, when the broadcast channel is under an input covariance constraint.
In particular, in the next section we
show that this equivalence holds for matrix power constraints that
satisfy a certain property, and we derive the linear precoders that
achieve optimal performance.  Section~\ref{secIV} then uses these
results to derive a sub-optimal algorithm for the case of the average
power constraint.

\section{Optimality of Linear Precoding for BC Secrecy}\label{secIII}

In this section we answer the following questions:
\begin{enumerate}[(a)]
\item For a given general MIMO Gaussian BC described by (\ref{lin1}),
where each node has an arbitrary number of antennas and the channel
input is under the covariance constraint (\ref{lin3}), is there any
$\bS\succeq\b0$ for which linear precoding can attain the secrecy
capacity region?
\item If yes, how can such $\bS$ be described?
\item For such $\bS$, what is the optimal linear precoder that allows
the rectangular S-DPC capacity region given by~(\ref{lin13}) to be
achieved?
\item If $\bS$ does not satisfy the condition for optimal linear
precoding in (a), what is the worst-case loss in secrecy capacity
incurred by using the linear precoding approach described in (b)
anyway?
\end{enumerate}
To begin, we give the following theorem as an answer to questions
(a) and (b) above.

\begin{thm}\label{lin_thm1}
Suppose the matrix power constraint $\bS\succeq\b0$ on the input covariance
$\bQ$ in~(\ref{lin3}) leads to generalized eigenvectors
in~(\ref{lin10}) that satisfy
$\mathrm{span}\{\bC_1\}\perp\mathrm{span}\{\bC_2\}$,
i.e. $\bC_1^H\bC_2=\b0$.  Then the secrecy capacity region
$\mathcal{C}_s(\bH,\bG,\bS)$ can be achieved with
$X=V_1+V_2$, where $V_1$ and $V_2$ are \emph{independent}
Gaussian precoders respectively corresponding to
$W_1$ and $W_2$, with zero means and covariance matrices $\bK^*_t$ and
$\bS-\bK^*_t$, with $\bK^*_t$ defined in~(\ref{lin8}).
\end{thm}
\begin{proof}
Recall that for any $\bS\succeq\b0$, the secrecy capacity region
$\mathcal{C}_s(\bH,\bG,\bS)$ is rectangular, so we only need to show
that when $\bC_2^H\bC_1=\b0$, the linear precoders $V_1$ and
$V_2$ characterized in this theorem are capable of achieving the
corner point $(R^*_1,R^*_2)$ given by~(\ref{lin13}).
From (\ref{lin4}), the achievable secrecy rate $R_1$ is given by
\begin{align}
R_1&=I(V_1;Y_1)-I(V_1;V_2)-I(V_1;Y_2|V_2) =I(V_1;Y_1)-I(V_1;Y_2|V_2)  \label{lin16}\\
&=I(V_1;\,\bH(V_1+V_2)+Z_1) -I(V_1;\,\bG(V_1+V_2)+Z_2|V_2)\nonumber\\
&=I(V_1;\,\bH(V_1+V_2)+Z_1) -I(V_1;\,\bG V_1+Z_2)\label{lin17}\\
&=\log\left|\mathbf{\Lambda}_1\right|=R^*_1 \; , \label{lin18}
\end{align}
where (\ref{lin16}) and the second part of (\ref{lin17}) come from
the fact that $V_1$ and $V_2$ are independent. Equation~(\ref{lin18})
is proved in Appendix A.  One can similarly show that
$R_2=R_2^*=-\log\left|\mathbf{\Lambda}_2\right|$ is achievable to
complete the proof.
\end{proof}

Theorem~\ref{lin_thm1} shows that the secrecy capacity region
corresponding to any $\bS$ with orthogonal $\bC_1$ and $\bC_2$ can be
achieved using either linear independent precoders $V_1$ and $V_2$, as
defined in Theorem \ref{lin_thm1}, or using the S-DPC approach, as
given by (\ref{lin6}). The next theorem expands on the answer to question (b) above, and also addresses (c).  First however we present the following lemma which holds for any $\bS\succeq\b0$.

\begin{lem}\label{lin_lem1}
For a given BC under the matrix power constraint (\ref{lin3}), for any
$\bS\succeq\b0$ we have $\mathrm{rank}(\bC_1)\le m$,
where $m$ is the number of positive eigenvalues of the matrix
$\bH^H\bH-\bG^H\bG$.
\end{lem}
\begin{proof}
Please see Appendix B.
\end{proof}

The following theorem presents a more specific condition on $\bS$
that results in generalized eigenvectors that satisfy
$\mathrm{span}\{\bC_1\}\perp\mathrm{span}\{\bC_2\}$.
\begin{thm}\label{lin_thm2}
For any $\bS\succeq\b0$,
the generalized eigenvectors $\bC_1$ and $\bC_2$ in~(\ref{lin10})
are orthogonal \emph{iff} there exists a matrix $\bT\in
\mathbb{C}^{n_t \times n_t}$ such that $\bS=\bT\bT^H$ and $\bT$
simultaneously block diagonalizes $\bH^H\bH$ and $\bG^H\bG$:
\begin{equation}\label{linshrink1}
\bT^H\bH^H\bH\bT = \left[\begin{array}{ccc}\bK_{\bH1} & \b0\\
\b0 & \bK_{\bH2}\end{array}\right]  \qquad
\bT^H\bG^H\bG\bT = \left[\begin{array}{ccc}\bK_{\bG1} & \b0\\
\b0 & \bK_{\bG2}\end{array}\right] \;,
\end{equation}
where the $m\times m$ matrices
$\bK_{\bH1}\succeq\b0$ and $\bK_{\bG1}\succeq\b0$ satisfy $\bK_{\bH1}
\succeq\bK_{\bG1}$ and $\bK_{\bH2} \preceq\bK_{\bG2}$.
\end{thm}
\begin{proof}
The proof begins by noting that if $\bS=\bT\bT^H$, then the pencil
in~(\ref{lin9}) and
\begin{align*}
\left(\bT^H\bH^H\bH\bT+\bI\;,\; \bT^H\bG^H\bG\bT+\bI\right)
\end{align*}
have exactly the same generalized eigenvalue matrix
$\mathbf{\Lambda}$, and thus the same secrecy capacity regions. The
remainder of the proof can be found in Appendix C.
\end{proof}

While algorithms exist to find $\bT$ that jointly block
diagonalizes $\bH^H\bH$ and $\bG^H\bG$ (see for example \cite{JBD} and
references therein), as mentioned in Appendix~C only those $\bT$ that
lead to $\bK_{\bH1} \succeq\bK_{\bG1}$ and $\bK_{\bH2}
\preceq\bK_{\bG2}$ are acceptable.  Later, we will
demonstrate that for any BC there are an infinite number of matrix 
constraints $\bS$ that can
achieve such a block diagonalization and hence allow for an optimal
linear precoding solution.

To conclude this section, we now answer question~(d) posed above.
Define the projection matrices $\bP_{\bC_i}=\bC_i(\bC_i^H\bC_i)^{-1}
\bC^H_i$ and $\bP^\perp_{\bC_i}=\bI-\bP_{\bC_i}$, and note that in
general, equation~(\ref{lin8}) is equivalent to
$\bK_t^*=\bS^{\frac{1}{2}}\bP_{\bC_1}\bS^{\frac{1}{2}}$ and
$\bS-\bK_t^*=\bS^{\frac{1}{2}}\bP_{\bC_1}^\perp\bS^{\frac{1}{2}}$.
When $\bC_1^H\bC_2 = \b0$, the optimal covariance matrices
for $V_1$ and $V_2$ also satisfy
\begin{align}
\bK_t^* & =  \bS^{\frac{1}{2}}\bP_{\bC_2}^\perp\bS^{\frac{1}{2}} \\
\bS-\bK_t^* & = \bS^{\frac{1}{2}}\bP_{\bC_2}\bS^{\frac{1}{2}} \; .
\end{align}
The following theorem explains the loss in secrecy that results when
linear precoding with these covariances is used for a matrix constraint 
$\bS$ that does {\em not} satisfy $\bC_1^H\bC_2 = \b0$.

\begin{thm}\label{lin_lem2}
Assume a linear precoding scheme $X=V_1+V_2$ for independent Gaussian
precoders $V_1$ and $V_2$ with zero means and covariance matrices
$\bK_t=\bS^{\frac{1}{2}}\bP^\perp_{\bC_2}\bS^{\frac{1}{2}}$ and
$\bS-\bK_t=\bS^{\frac{1}{2}}\bP_{\bC_2}\bS^{\frac{1}{2}}$,
respectively.  Also define
$\bN=\left(\bC_2^H\bP^\perp_{\bC_1}\bC_2\right)^{-1}
\bC_2^H\bP^\perp_{\bC_1}\bP^\perp_{\bC_2}\bC_1$.  The loss in secrecy
capacity that results from using this approach in the two-user BC is
\emph{at most} $\log\left|\bI+\bN^H\bN\right|$ for each user.  In
particular, the following secrecy rate pair is achievable:
\begin{align}\label{lin19}
\begin{split}
&R_1=\max(0,\, R_1^*-\log\left|\bI+\bN^H\bN\right|) =\max(0,\, \log\left|\mathbf{\Lambda}_1\right|-\log\left|\bI+\bN^H\bN\right|)\\
&R_2= \max(0,\, R_2^*-\log\left|\bI+\bN^H\bN\right|) =\max(0,\,-\log\left|\mathbf{\Lambda}_2\right|-\log\left|\bI+\bN^H\bN\right|)  \;.
\end{split}
\end{align}
\end{thm}
\begin{proof}
See Appendix D.
\end{proof}

\begin{rem}\label{lin_rem2}
Note that if $\bC_1$ and $\bC_2$ are orthogonal, then $\bN=\b0$ and
$R_i=R^*_i$, $i=1,2$, is achievable, as discussed in Theorem~1.
\end{rem}

\section{Sub-Optimal Solutions Under an Average Power Constraint}\label{secIV}

So far we have shown that if the broadcast channel (\ref{lin1}) is
under the matrix power constraint $\bS$ (\ref{lin3}), then linear
precoding as defined by Theorem \ref{lin_thm1} is an optimal solution
when $\bS$ satisfies the condition described in Theorem
\ref{lin_thm2}.  In the following we propose a suboptimal closed-form
linear precoding scheme for the general MIMO Gaussian BC under the
\emph{average} power constraint (\ref{lin2}), where as mentioned
earlier there exists no optimal closed-form solution that
characterizes the secrecy capacity region.  We begin with some
preliminary results, then we develop the algorithm for the general
MIMO case, and finally we present an alternative algorithm
specifically for the MISO case since it offers additional insight.

\subsection{Preliminary Results}

\begin{rem}\label{lin_remetx1}
Suppose that the input covariance matrix $\bQ$ leads to a point on the
Pareto boundary of the secrecy capacity region given by~(\ref{lin14})
under the average power constraint~(\ref{lin2}).  Then
$\mathrm{Tr}(\bQ)=P_t$ and $\bQ$ cannot have any component in the
nullspace of $\bH^H\bH+\bG^H\bG$, and thus
$\mathcal{C}_s(\bH,\bG,P_t)=\mathcal{C}_s(\bH_{eq},\bG_{eq},P_t)$,
where $\bH_{eq}=\bH\bU_p$, $\bG_{eq}=\bG\bU_p$ and $\bU_p$ contains
the singular vectors corresponding to the non-zero singular values of
$\bH^H\bH+\bG^H\bG$.
\end{rem}

According to Remark \ref{lin_remetx1}, we can assume without loss of
generality that $\bH^H\bH+\bG^H\bG$ is full-rank; otherwise, we could
replace $\bH,\bG$ with $\bH_{eq},\bG_{eq}$ and have an equivalent
problem where $\bH_{eq}^H\bH_{eq}+\bG_{eq}^H\bG_{eq}$ is full-rank and
the secrecy capacity region is the same (in such a case, $n_t$ would
then represent the number of transmitted data streams rather than the
number of antennas).  With this result, we have the following lemma.

\begin{lem}\label{lin_lem3}
Define
\begin{align}\label{lin20}
\bW=(\bH^H\bH+\bG^H\bG)^{-\frac{1}{2}}  \;.
\end{align}
Then $\bW\bH^H\bH\bW$ and $\bW\bG^H\bG\bW$ commute and hence
share the same set of eigenvectors:
\begin{align}\label{lin21}
\begin{split}
\bW\bH^H\bH\bW &=\bPhi_{\bw} \bSig_1 \bPhi_{\bw}^H\\
\bW\bG^H\bG\bW &=\bPhi_{\bw} \bSig_2 \bPhi_{\bw}^H \; ,
\end{split}
\end{align}
where $\bPhi_{\bw}$ is the (unitary) matrix of eigenvectors
and $\bSig_1 \succeq \b0, \bSig_2 \succeq \b0$ the
corresponding eigenvalues.
\end{lem}
\begin{proof}
See Appendix E.
\end{proof}

Without loss of generality, we assume that the
columns of $\bPhi_{\bw}$ are sorted such that the first $\rho$
diagonal elements of $\bSig_1$ are greater than the first $\rho$
diagonal elements of $\bSig_2$, and the last $n_t-\rho$ diagonal
elements of $\bSig_1$ are less than or equal to those of
$\bSig_2$. Recall from Lemma~\ref{lin_lem1} that $0\le\rho\le m$,
where $m$ is the number of positive eigenvalues of
$\bH^H\bH-\bG^H\bG$.  Thus,
\begin{equation} \label{lin22}
\bSig_1 = \left[\begin{array}{ccc}\bSig_{1\rho} & \b0\\\b0 &
\bSig_{1\bar{\rho}}\end{array}\right]  \qquad
\bSig_2 = \left[\begin{array}{ccc}\bSig_{2\rho} & \b0\\\b0 & \bSig_{2\bar{\rho}}\end{array}\right]
\end{equation}
where $\bSig_{i\rho}$ is $\rho\times \rho$, $\bSig_{i\bar{\rho}}$ is
$(n_t-\rho)\times (n_t-\rho)$, $\bSig_{1\rho} \succ\bSig_{2\rho}$ and
$\bSig_{1\bar{\rho}} \preceq\bSig_{2\bar{\rho}}$.

Now define
\begin{align}
\bS_{\bw} & =  \bT_{\bw} \bT_{\bw}^H=\bW\bPhi_{\bw}\;\bP \;\bPhi^H_{\bw} \bW
\label{lin24} \\
\bT_{\bw} & =  \bW\bPhi_{\bw}\bP^{\frac{1}{2}} \; ,
\label{lin23}
\end{align}
where $\bW$ and $\bPhi_{\bw}$ are given in Lemma~\ref{lin_lem3} and
$\bP\succeq\b0$ is any block-diagonal matrix partitioned in the same way
as $\bSig_1$ and $\bSig_2$.  With these definitions, we see
from~(\ref{lin21}) that $\bT^H_{\bw}\bH^H\bH\bT_{\bw}$ and
$\bT^H_{\bw}\bG^H\bG\bT_{\bw}$ are block diagonal. Thus, from Theorem
\ref{lin_thm2}, a BC with the matrix power constraint
$\bS_{\bw}=\bT^H_{\bw}\bT_{\bw}$ leads to a matrix pencil
$\left(\bS_{\bw}^{\frac{1}{2}}\bH^H\bH\bS_{\bw}^{\frac{1}{2}}+\bI\;,\;
\bS_{\bw}^{\frac{1}{2}}\bG^H\bG\bS_{\bw}^{\frac{1}{2}}+\bI\right)$
with generalized eigenvectors $\bC_\bw=[\bC_{1\bw} \; \bC_{2\bw}]$
that satisfy $\bC_{1\bw}^H\bC_{2\bw}=\b0$, where $\bC_{1\bw},\bC_{2\bw}$
correspond to generalized eigenvalues that are larger or
less-than-or-equal-to one, respectively.

\begin{rem}\label{lin_remetx2}
Since the above result holds for any block-diagonal $\bP\succeq 0$ with appropriate dimensions, then for every BC there are an infinite number of matrix power constraints $\bS_{\bw}$ that achieve a block diagonalization and hence allow for an \emph{optimal} linear precoding solution.
\end{rem}

In the following, we restrict our attention to diagonal rather than
block-diagonal matrices $\bP$, for which a closed form solution can be
derived. From Theorem~\ref{lin_thm1}, we have the following result.

\begin{lem}\label{lin_lem4}
For any diagonal $\bP \succeq 0$, the secrecy capacity of the broadcast channel in~(\ref{lin1}) under
the matrix power constraint $\bS_{\bw}= \bT_{\bw} \bT_{\bw}^H$ defined
in~(\ref{lin20})-(\ref{lin23}) can be obtained by linear precoding.
In particular,
\begin{align}\label{lin25}
X=\bW\bPhi_{\bw} \left[\begin{array}{ccc} V'_1 \\ V'_2
\end{array}\right] = V_1 + V_2
\end{align}
where $V'_1\in\mathbb{C}^{\rho}$ and $V'_2\in\mathbb{C}^{n_t-\rho}$
are independent Gaussian random vectors with zero
means and covariance matrices $\bP_1$ and $\bP_2$ such that
\begin{align}\label{lin26}
\bP=\left[\begin{array}{ccc}\bP_1 & 0\\0 & \bP_2\end{array}\right] \; ,
\end{align}
and as before $V_1, V_2$ represent independently encoded Gaussian
codebook symbols corresponding to the confidential messages $W_1$ and
$W_2$, with zero means and covariances $\bK_{t\bw}^*$ and
$\bS_{\bw}-\bK_{t\bw}^*$ respectively given by
\begin{align}
\bK_{t\bw}^* & = \bW\bPhi_{\bw} \left[\begin{array}{cc} \bP_1 & \b0 \\ \b0
& \b0 \end{array} \right] \bPhi_{\bw}^H \bW \label{lin25a} \\
\bS_\bw - \bK_{t\bw}^* & = \bW\bPhi_{\bw}\left[ \begin{array}{cc} \b0 & \b0 \\ \b0
& \bP_2 \end{array} \right] \bPhi_{\bw}^H \bW \; . \label{lin25b}
\end{align}
\end{lem}
\begin{proof}
The matrix $\bT_{\bw}$ simultaneously block diagonalizes $\bH^H\bH$
and $\bG^H\bG$, so by Theorems~\ref{lin_thm1} and~\ref{lin_thm2} we
know that linear precoding can achieve the secrecy capacity region.
The proof is completed in Appendix F by showing the equality
in~(\ref{lin25}), and showing that~(\ref{lin25a}) corresponds to the optimal
covariance in~(\ref{lin8}).
\end{proof}

From the proof in Appendix F and~(\ref{lin21})-(\ref{lin23}), we see that under the matrix power constraint $\bS_{\bw}$ given by (\ref{lin24}) with diagonal $\bP$, the general BC is transformed to an equivalent BC with a set of parallel independent
subchannels between the transmitter and the receivers, and it suffices
for the transmitter to use independent Gaussian codebooks across these
subchannels.  In particular, the diagonal entries of $\bP_1$ and $\bP_2$ represent the power assigned to these independent subchannels prior to application of the precoder $\bW\bPhi_{\bw}$ in~(\ref{lin23})\footnote{Note that the matrices $\bP_1,\bP_2$ do not represent the actual transmitted power, since the columns of $\bW\bPhi_{\bw}$ are not unit-norm.}.  
From (\ref{lin25}), the signals at the two receivers are given by
\begin{eqnarray*}
\by_1 & = & \bH\bW\bPhi_{\bw} \left[ \begin{array}{c} \bv'_1 \\ \bv'_2 \end{array} \right] + \bz_1 \\
& = & \bGam_1 \bSig_1 \left[ \begin{array}{c} \bv'_1 \\ \bv'_2 \end{array} \right] + \bz_1 \\
& = & \bGam_1 \left[ \begin{array}{c} \bSig_{1\rho} \bv'_1 \\ \bSig_{1\bar{\rho}} \bv'_2 \end{array} \right] + \bz_1 \\
\by_2 & = & \bGam_2 \left[ \begin{array}{c} \bSig_{2\rho}  \bv'_1 \\ \bSig_{2\bar{\rho}}  \bv'_2 \end{array} \right] + \bz_2 \; ,
\end{eqnarray*}
where $\bGam_1,\bGam_2$ are unitary.  
The confidential message for receiver~1 is thus transmitted with power loading $\bP_1$ over those subchannels which are degraded for receiver~2 ($\bSig_{1\rho} \succ \bSig_{2\rho}$), while receiver~2's confidential message has power loading $\bP_2$ over subchannels which are degraded for receiver~1 ($\bSig_{2\bar{\rho}} \succ \bSig_{1\bar{\rho}}$).  Any subchannels for which the diagonal elements of $\bSig_{2\bar{\rho}}$ are equal to those of $\bSig_{1\bar{\rho}}$ are useless from the viewpoint of secret communication, but could be used to send common non-confidential messages.

From Theorem \ref{lin_thm1}, the rectangular secrecy capacity region
of the MIMO Gaussian BC (\ref{lin1}) under the matrix power
constraint $\bS_{\bw}$ (\ref{lin24}) is defined by the corner points
\begin{align}\label{lin27}
\begin{split}
R^*_1(\bP_1)&= \log\left|\mathbf{\Lambda}_{1\bw}\right|= \log\left|\bI+\bSig_{1\rho}\bP_1\right|-\log\left|\bI+\bSig_{2\rho}\bP_1\right|   \\
R^*_2(\bP_2)&= -\log\left|\mathbf{\Lambda}_{2\bw}\right|= \log\left|\bI+\bSig_{2\bar{\rho}}\bP_2\right|-\log\left|\bI+\bSig_{1\bar{\rho}}\bP_2\right| \; ,
\end{split}
\end{align}
where $\mathbf{\Lambda}_{i\bw}$ is given by (\ref{linap35}) in
Appendix~F.  Note that we have explicitly
written $R^*_1$ as a function of the diagonal matrix $\bP_1\succeq\b0$
to emphasize that $\bP_1$ contains the only parameters that can be optimized
for $R^*_1$. More precisely, since for a given matrix power
constraint $\bS_{\bw}$, $\bSig_{1\rho}$ and $\bSig_{2\rho}$ are channel dependent and thus fixed, as shown in (\ref{lin21})-(\ref{lin22}).  A similar
description is also true for $R^*_2$.

\subsection{Algorithm for the MIMO Case Under the Average Power Constraint}\label{sec:mimo}

Here we propose our sub-optimal closed form solution based on linear
precoding for the broadcast channel under the \emph{average} power
constraint (\ref{lin2}). The goal is to find the diagonal matrix $\bP$ in (\ref{lin24}) that maximizes $R_i^*$ in~(\ref{lin27}) for a given allocation of the transmit power to message $W_i$, and that satisfies the average power constraint\footnote{ Note that since we want to characterize the achievable secrecy rate points on the Pareto boundary, we use an equality constraint on the total power $P_t$ in~(\ref{lin28}).}
\begin{align}\label{lin28}
\mathrm{Tr}(E\{XX^H\})&=\mathrm{Tr}(\bS_{\bw})= \mathrm{Tr}(\bW\bPhi_{\bw}\bP\bPhi_{\bw}^H \bW) \nonumber\\
&= \mathrm{Tr}(\bPhi_{\bw}^H \bW^2\bPhi_{\bw}\bP)= \mathrm{Tr}\left(\bPhi_{\bw}^H (\bH^H\bH+\bG^H\bG)^{-1}\bPhi_{\bw}\;\bP\right)= P_t\; .
\end{align}
Noting that $\bPhi_{\bw}$ can be written as $\bPhi_{\bw}=[\bPhi_{1\bw}
\; \bPhi_{2\bw}]$, where $\bPhi_{1\bw}$ is a $n_t\times\rho$
submatrix corresponding to the eigenvalues in $\bSig_{1\rho}$,
(\ref{lin28}) can be rewritten as
\begin{align}\label{lin29}
\mathrm{Tr}(E\{XX^H\})&=\mathrm{Tr}\left(\bPhi_{\bw}^H (\bH^H\bH+\bG^H\bG)^{-1}\bPhi_{\bw}\;\bP\right) \nonumber\\
&= \mathrm{Tr}\left(\bPhi_{1\bw}^H (\bH^H\bH+\bG^H\bG)^{-1}\bPhi_{1\bw}\;\bP_1\right) +\mathrm{Tr}\left(\bPhi_{2\bw}^H (\bH^H\bH+\bG^H\bG)^{-1}\bPhi_{2\bw}\;\bP_2\right)  \nonumber\\
&=  \mathrm{Tr}\left(\bA_1\bP_1\right) +\mathrm{Tr}\left(\bA_2\bP_2\right)=P_t
\end{align}
where we defined positive definite matrices $\bA_i=\bPhi_{i\bw}^H (\bH^H\bH+\bG^H\bG)^{-1}\bPhi_{i\bw}$, $i=1,2$.

Our sub-optimal closed-form solution for the BC under the average
power constraint (\ref{lin2}) is not optimal, since instead of doing
an exhaustive search over all $\bS\succeq\b0$ with
$\mathrm{Tr}(\bS)=P_t$ as indicated in (\ref{lin14}), we will only
consider specific $\bS$ matrices of the form given for $\bS_{\bw}$ in
(\ref{lin24}) with diagonal $\bP$.
Since
$R_i^*(\bP_i)$ is only a function of $\bP_i$, $R^*_1(\bP_1)$ and
$R^*_2(\bP_2)$ can be optimized separately for any power fraction
$\alpha$ ($0\le\alpha\le1$) under the constraints
$\mathrm{Tr}\left(\bA_1\bP_1\right)=\alpha P_t$ and
$\mathrm{Tr}\left(\bA_2\bP_2\right)=(1-\alpha) P_t$, respectively.

\begin{thm}\label{lin_lem5}
For any  $\alpha$, $0\le\alpha\le1$, the diagonal elements of the
optimal $\bP^*_1$ and $\bP^*_2$ are given by
\begin{align}
p^*_{1i} & = \max\left(0,\frac{-(\sigma_{1\rho i}+\sigma_{2\rho i})+\sqrt{(\sigma_{1\rho i}-\sigma_{2\rho i})^2+4(\sigma_{1\rho i}-\sigma_{2\rho i}) \sigma_{1\rho i}\sigma_{2\rho i}/(\mu_1 a_{1i})}}{2\, \sigma_{1\rho i}\sigma_{2\rho i}}\right)
\label{lin30} \\
p^*_{2i} & = \max\left(0,\frac{-(\sigma_{1\bar{\rho} i}+\sigma_{2\bar{\rho} i})+\sqrt{(\sigma_{2\bar{\rho} i}-\sigma_{1\bar{\rho} i})^2+4(\sigma_{2\bar{\rho} i}-\sigma_{1\bar{\rho} i}) \sigma_{2\bar{\rho} i}\sigma_{1\bar{\rho} i}/(\mu_2 a_{2i})}}{2\, \sigma_{1\bar{\rho} i}\sigma_{2\bar{\rho} i}}\right)
\label{lin31} \;,
\end{align}
where $\sigma_{1\rho i}$, $\sigma_{2\rho i}$, and $a_{1i}$ are the
$i^{\rm th}$ diagonal elements of $\bSig_{1\rho}$, $\bSig_{2\rho}$,
and $\bA_1$, respectively, where $0\le i \le\rho$. Also
$\sigma_{1\bar{\rho} i}$, $\sigma_{2\bar{\rho} i}$, and $a_{2i}$ are the
$i^{\rm th}$ diagonal elements of $\bSig_{1\bar{\rho}}$, $\bSig_{2\bar{\rho}}$,
and $\bA_2$, respectively, where $0\le i \le(n_t-\rho)$.
The Lagrange parameters $\mu_1>0$ and $\mu_2>0$ are chosen to
satisfy the average power constraints
$\mathrm{Tr}\left(\bA_1\bP_1\right)=\alpha P_t$ and
$\mathrm{Tr}\left(\bA_2\bP_2\right)=(1-\alpha) P_t$, respectively.
\end{thm}
\begin{proof}
We want to optimize diagonal matrices $\bP_1$ and $\bP_2$ so that the secrecy rates $R^*_1(\bP_1)$ and $R^*_2(\bP_2)$, given by (\ref{lin27}), are maximized for a given $\alpha$, $0\le\alpha\le1$. Since $R^*_i(\bP_i)$ only depends on $\bP_i$, the two terms in~(\ref{lin27}) can be maximized independently.  We show the result for $i=1$; the procedure for $i=2$ is identical.  From (\ref{lin27}), the Lagrangian associated with $\mathrm{max}_{\mathrm{Tr}(\bA_1\bP_1)=\alpha P_t}\;R^*_1(\bP_1)$ is
\begin{align}\label{eq15}
\mathcal{L}&=\log\left|\bI+\bSig_{1\rho}\bP_1\right|-\left|\bI+\bSig_{2\rho}\bP_1\right| -\mu_1 \mathrm{Tr}(\bA_1\bP_1)\nonumber\\
&=\sum_{i} [\log(1+\sigma_{1\rho i}p_{1i}) - \log(1+\sigma_{2\rho i}p_{1i})] -\mu_1 \sum_{i} a_{1i} \,p_{1i} \;,
\end{align}
where $\mu_1>0$ is the Lagrange multiplier. Since $\bSig_{1\rho}\succ\bSig_{2\rho}$, Eq. (\ref{eq15}) represents a convex optimization problem. The optimal $\bP_1^*$ with diagonal elements given by (\ref{lin30}) is simply obtained by applying the KKT conditions to
(\ref{eq15}).
\end{proof}

\begin{cor}\label{lin_cor1}
 For any $\alpha$, $0\le\alpha\le1$, let $R^*_1(\alpha)$ and $R^*_2(\alpha)$ represent the corner points given
by (\ref{lin27}) for the optimal $\bP_1^*$ and $\bP_2^*$, given by
(\ref{lin30}) and (\ref{lin31}). The achievable secrecy rate region of
the above approach under the average power constraint (\ref{lin2}) is
the convex hull of all obtained corner points and is given by
\begin{align}\label{lin32}
\mathcal{R}_s(\bH,\bG,P_t)=\bigcup_{0\le\alpha\le1} \left(R^*_1(\alpha)\, , \, R^*_2(\alpha)\right) \;.
\end{align}
\end{cor}

It is interesting to note that, unlike the conventional broadcast channel without secrecy constraints where uniform power allocation is optimal in maximizing the sum-rate in the high SNR regime \cite{Jindal}, the high SNR power allocation for the BC with confidential messages is a special form of waterfilling as described in the following lemma.
\begin{lem}\label{lin_lemext1}
For high SNR $(P_t\rightarrow \infty)$, the asymptotic optimal power allocations given by (\ref{lin30})-(\ref{lin31}) are
\begin{align}
p^*_{1i} & = \sqrt{\frac{1}{\mu_1 a_{1i}}\left(\frac{1}{\sigma_{2\rho i}}-\frac{1}{\sigma_{1\rho i}}\right)} \label{linext1} \\
p^*_{2i} & =\sqrt{\frac{1}{\mu_2 a_{2i}} \left(\frac{1}{\sigma_{1\bar{\rho} i}}-\frac{1}{\sigma_{2\bar{\rho} i}}\right)}\label{linext2} \; .
\end{align}
\end{lem}
\begin{proof}
To show (\ref{linext1}) we note that $\mu_1\rightarrow 0$ when $P_t\rightarrow \infty$. Thus (\ref{lin30}) can be written as
$$p^*_{1i} = \frac{\sqrt{4(\sigma_{1\rho i}-\sigma_{2\rho i}) \sigma_{1\rho i}\sigma_{2\rho i}/(\mu_1 a_{1i})}}{2\, \sigma_{1\rho i}\sigma_{2\rho i}}= \sqrt{\frac{1}{\mu_1 a_{1i}}\;\frac{\sigma_{1\rho i}-\sigma_{2\rho i}}{\sigma_{1\rho i}\sigma_{2\rho i}}}= \sqrt{\frac{1}{\mu_1 a_{1i}}\left(\frac{1}{\sigma_{2\rho i}}-\frac{1}{\sigma_{1\rho i}}\right)}.$$ (\ref{linext2}) is proved similarly.
\end{proof}

It is also worth noting that the solution in~(\ref{lin30})-(\ref{lin31}) approaches the standard point-to-point MIMO waterfilling solution when one of the channels is dominant.  For example, let $\bG\rightarrow 0$.  We will show that the optimal input covariance simplifies to the waterfilling solution for $\bH$, given by $\bQ^*_P=\bPhi_{H}\; (\frac{1}{\mu}\bI-\mathbf{\Lambda}^{-1}_H)^+
\;\bPhi_{H}^H$, where unitary $\bPhi_{H}$ and diagonal $\mathbf{\Lambda}_H$ are obtained from the eigenvalue decomposition $\bH^H\bH=\bPhi_{H}
\mathbf{\Lambda}_H\bPhi^H_{H}$.  The capacity of the point-to-point MIMO Gaussian link is
\begin{align} \label{P2P}
C= \log\left|\bI+\bH\,\bQ^*_P\bH^H\right|=\log\left|\bI + \mathbf{\Lambda}_H \left(\frac{1}{\mu}\bI - \mathbf{\Lambda}_H^{-1}\right)^+\right|.
\end{align}
When $\bG\rightarrow\b0$, we note from (\ref{lin20}) and (\ref{lin21}) that  $\bPhi_{\bw}\rightarrow\bPhi_H$, $\bSig_1=\bSig_{1\rho}\rightarrow\bI$, $\bSig_2\rightarrow\b0$, and $\bP_1=\bP$.  Consequently, $R^*_2\rightarrow0$ and $R^*_1\rightarrow\log\left|\bI+\bP^*\right|$, where $\bP^*$ is a diagonal matrix with diagonal elements given by  (\ref{lin30}).  The average power constraint in (\ref{lin29}) becomes $\mathrm{Tr}\left(\bA\bP\right)=P_t$, where $\bA=\bPhi_{\bw}^H (\bH^H\bH+\bG^H\bG)^{-1}\bPhi_{\bw}\rightarrow \mathbf{\Lambda}^{-1}_H$ when $\bG\rightarrow\b0$. Thus the $i^{\rm th}$ diagonal element of $\bA$ converges to the $i^{\rm th}$ diagonal element of $\mathbf{\Lambda}^{-1}_H$. Starting from~(\ref{lin30}) and applying L'H\^{o}pital's rule, when $\bG\rightarrow 0$ and hence $\sigma_{2\rho i}\rightarrow 0$, we have $\bP^*\rightarrow \mathbf{\Lambda}_H (\frac{1}{\mu}\bI-\mathbf{\Lambda}_H^{-1})^+$, and consequently,
$$\lim_{\bG \rightarrow\b0}
R^*_1=\log\left|\bI+\bP^*\right|= \log\left|\bI + \mathbf{\Lambda}_H \left(\frac{1}{\mu}\bI - \mathbf{\Lambda}_H^{-1}\right)^+\right| \; .
$$

\subsection{Alternative Approach for the MISO Case}

Here we focus on the BC in (\ref{lin1}) for the MISO case under an
average power constraint, where both receivers have a single antenna
and the transmitter has $n_t\ge 2$ antennas:
\begin{align*}
\begin{split}
y_1=\bh^H\bx+ z_1\\
y_2=\bg^H\bx+ z_2 \; ,
\end{split}
\end{align*}
where the channels are represented by the $n_t \times 1$ vectors $\bh$
and $\bg$.  The MISO case is the only BC scenario whose secrecy
capacity region under the \emph{average} power constraint~(\ref{lin2})
is characterized in closed-form. In particular, it was shown in
\cite{Liu09} that
\begin{align}\label{lin33}
\mathcal{C}_s(\bh,\bg,P_t) = \bigcup_{0\le\alpha\le1} \left(C_1(\alpha), C_2(\alpha)\right)
\end{align}
where $\left(C_1(\alpha), C_2(\alpha)\right)$ is the secrecy rate pair
on the Pareto boundary of the secrecy capacity region for the
power fraction $\alpha$, $0\le\alpha\le1$, where power $\alpha P_t$ is
allocated to receiver~1's message and $(1-\alpha)P_t$ is allocated to
receiver~2's message. Furthermore, we have
\cite{Liu09}
\begin{align}\label{lin34}
\begin{split}
C_1(\alpha)&=\log\gamma_1(\alpha)\\
C_2(\alpha)&=\log \gamma_2(\alpha) \; ,
\end{split}
\end{align}
where
$$\gamma_1(\alpha)= \frac{1+\alpha
P_t\;\be_1^H\bh\bh^H\be_1}{1+\alpha P_t\;\be_1^H\bg\bg^H\be_1} \; , $$
$\be_1$ is the unit length principal generalized eigenvector of
$(\bI+P_t\bh\bh^H\;;\;\bI+P_t\bg\bg^H)$, $\gamma_2(\alpha)$
is the largest generalized eigenvalue of
$$(\bI+\frac{(1-\alpha)P_t}{1+\alpha P_t
|\be_1^H\bg|}\bg\bg^H\;;\;\bI+\frac{(1-\alpha)P_t}{1+\alpha P_t
|\be_1^H\bh|}\bh\bh^H) \; ,$$
and $\be_2$ denotes the unit length
generalized eigenvector corresponding to $\gamma_2(\alpha)$.
Note that the achievablity of (\ref{lin34}) is still based on S-DPC.

While we could have just used the results of Section~\ref{sec:mimo}
for the MISO case, we will see that the advantage of considering a
different approach here is that we obtain a more succinct expression for the
achievable secrecy rate region for linear precoding, and we are able
to quantify the loss in secrecy rate incurred by linear precoding
under the average power constraint compared with
$(C_1(\alpha),C_2(\alpha))$.  This was not possible in the MIMO
case.

Referring to (\ref{lin6}), it was shown in \cite{Liu09} that for the
secrecy rate pair given by (\ref{lin34}), $U_1$ and $U_2$ have
covariance matrices $\alpha P_t \be_1\be_1^H$ and $(1-\alpha) P_t
\be_2\be_2^H$, respectively. Thus, the specific input covariance
matrix that attains~(\ref{lin34}) is given by
\begin{align}\label{lin35}
\bS_Q=\alpha P_t \be_1\be_1^H +(1-\alpha) P_t \be_2\be_2^H   \; ,
\end{align}
where $\mathrm{Tr}(\bS_Q)=P_t$ and $\mathrm{rank}(\bS_Q)=2$.
Equivalently, one can say that under the \emph{matrix} power
constraint $\bS_Q$, the corner point of the corresponding rectangular
secrecy capacity region is given by (\ref{lin34}).  The union of these
corner points constructs the Pareto boundary of the secrecy capacity
region under the \emph{average} power constraint, where any point on
the boundary is given by~(\ref{lin34}) for a different $\alpha$ and is
achieved under the matrix power constraint $\bS_Q$ given
by~(\ref{lin35}).

Using the above fact, we now present a different linear precoding
scheme as an alternative to Corollary
\ref{lin_cor1} for the MISO BC under the average power
constraint (\ref{lin2}).
\begin{cor}\label{lin_cor2}
Using the linear precoding scheme proposed in Theorem~\ref{lin_lem2}
for the MISO BC under an average power constraint, the following secrecy
rate region is achievable:
$$\mathcal{R}_s(\bh,\bg,P_t) = \bigcup_{0\le\alpha\le1} \left(R_1(\alpha), R_2(\alpha)\right)\;,$$
where
\begin{align}\label{lin36}
\begin{split}
&R_1(\alpha)=\max(C_1(\alpha)-\log\left(1+ (\bc_2^H\bP^{\perp}_{\bc_1}\bc_2)^{-2}\; |\bc_1^H\bP^{\perp}_{\bc_2}\bP^{\perp}_{\bc_1}\bc_2|^2 \right),\, 0) \\
&R_2(\alpha)= \max(C_2(\alpha)-\log\left(1+ (\bc_2^H\bP^{\perp}_{\bc_1}\bc_2)^{-2}\; |\bc_1^H\bP^{\perp}_{\bc_2}\bP^{\perp}_{\bc_1}\bc_2|^2 \right),\, 0)  \;,
\end{split}
\end{align}
$C_1(\alpha)$ and $C_2(\alpha)$ are given by (\ref{lin34}),
\begin{align}
\bc_1 & =  \frac{1}{\sqrt{\be_1^H(\bS^{-1}_Q+\bg\bg^H)\be_1}}\; \bS_Q^{-\frac{1}{2}}\be_1
\label{lin41} \\
\bc_2 & =  \frac{1}{\sqrt{\mathbf{f}_1^H(\bS^{-1}_Q+\bg\bg^H)\mathbf{f}_1}}\;
\bS_Q^{-\frac{1}{2}}\mathbf{f}_1 \; , \label{lin42}
\end{align}
and where $\mathbf{f}_1$ is the unit length principal generalized
eigenvector of $(\bI+P_t\bg\bg^H\;;\;\bI+P_t\bh\bh^H)$.
\end{cor}

\begin{proof}
From Remark \ref{lin_remetx1}, and by noting that for any MISO BC,
$\bh\bh^H+\bg\bg^H$ has at most 2 non-zero eigenvalues, any MISO BC
can be modeled with a scenario involving just two transmit antennas.
Thus, without loss of generality, we assume that $n_t=2$. From
Theorem~\ref{lin_lem2}, we only need to characterize $\bc_1$ and
$\bc_2$, where $\bc_1$ ($\bc_2$) is the generalized eigenvector of the
pencil
\begin{align} \label{lin37}
\left(\bS_Q^{\frac{1}{2}}\bh\bh^H\bS_Q^{\frac{1}{2}}+\bI\;,\; \bS_Q^{\frac{1}{2}}\bg\bg^H\bS_Q^{\frac{1}{2}}+\bI\right)
\end{align}
corresponding to the generalized eigenvalue larger (less) than 1,
$\lambda_1$ ($\lambda_2$).

From (\ref{lin6}) and  (\ref{lin8}), the covariance matrix of $U_1$ can be rewritten as
\begin{align} \label{lin38}
\bK_t^* = \bS_Q^{\frac{1}{2}}\;[\bc_1\; \bc_2]\left[
\begin{array}{ccc}
(\bc_1^H\bc_1)^{-1} & 0\\
0 & 0
\end{array}
\right][\bc_1\; \bc_2]^H \bS_Q^{\frac{1}{2}}  = \frac{1}{\bc_1^H\bc_1}\,\bS_Q^{\frac{1}{2}} \bc_1\;\bc_1^H \bS_Q^{\frac{1}{2}} \; .
\end{align}
Comparing (\ref{lin38}) with the covariance matrix of $U_1$ reported
in \cite{Liu09}, we have $\alpha P_t
\be_1\be_1^H=\frac{1}{|\bc_1|^2}\,\bS_Q^{\frac{1}{2}} \bc_1\;\bc_1^H
\bS_Q^{\frac{1}{2}}$. This results in\footnote{Note that multiplication
by a factor $\mathrm{exp}(j\theta)$ is required for a precise equaltiy,
but since this term disappears in the final result, we
simply ignore it.}
\begin{align} \label{lin39}
\frac{\bc_1}{\|\bc_1\|}= \sqrt{\alpha P_t}\;\bS_Q^{-\frac{1}{2}}\be_1 \;.
\end{align}
On the other hand, from the definition of $\bc_1$ and $\bc_2$ (see
(\ref{lin10})-(\ref{lin11}) for example), we have
\begin{align}\label{lin40}
\begin{split}
&[\bc_1\; \bc_2]^H\left[\bS_Q^{\frac{1}{2}}\bh\bh^H\bS_Q^{\frac{1}{2}}+\bI\right][\bc_1\; \bc_2]=\left[\begin{array}{ccc}\lambda_1 & 0\\ 0& \lambda_2\end{array}\right] = \left[\begin{array}{ccc}\gamma_1(\alpha) & 0\\ 0& \gamma^{-1}_2(\alpha)\end{array}\right]\\
&[\bc_1\; \bc_2]^H\left[\bS_Q^{\frac{1}{2}}\bg\bg^H\bS_Q^{\frac{1}{2}}+\bI\right][\bc_1\; \bc_2]=\bI
\end{split}
\end{align}
\begin{figure}[t]
\centering
\includegraphics[width=3in,height=3in]{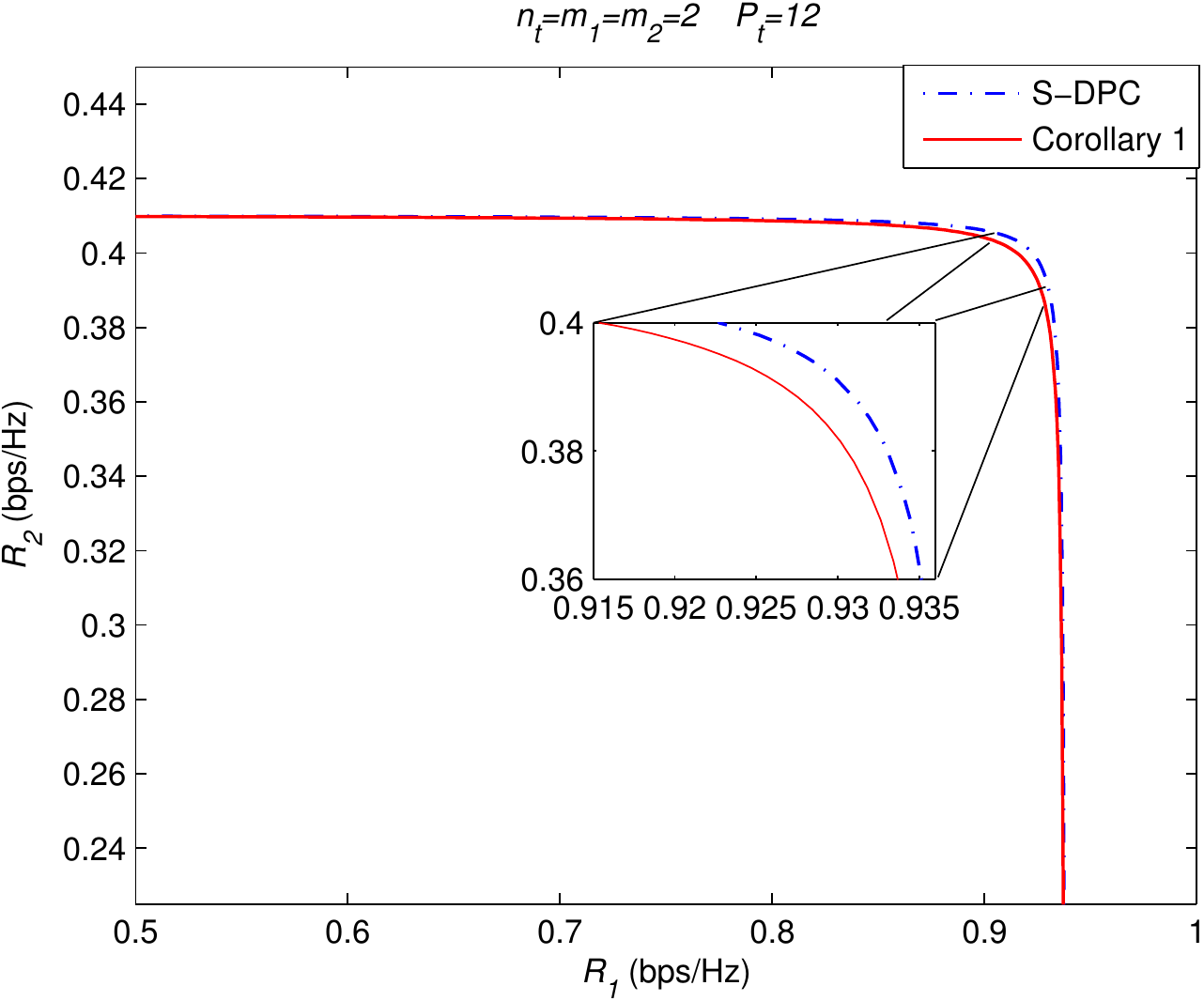}
\caption{Secrecy capacity region of S-DPC together with secrecy rate
region for linear precoding with $P_t=12$, $n_t=m_1=m_2=2$.}
\label{lin_exm1}
\end{figure}
where $\gamma_1(\alpha)$ and $\gamma_2(\alpha)$ are defined after
(\ref{lin34}), and the fact that $\lambda_1=\gamma_1(\alpha)$ and
$\lambda_2=\frac{1}{\gamma_2(\alpha)}$ comes from the argument after
(\ref{lin35}) and by comparing (\ref{lin13}) and (\ref{lin34}).
Substituting (\ref{lin39}) in (\ref{lin40}), after some simple
calculations, $\bc_1$ can be explicitly written as in~(\ref{lin41}).
Recalling that $\bc_1$ is the principal generalized eigenvector
of~(\ref{lin37}) and $\bc_2$, which corresponds to the smallest
generalized eigenvalue of the pencil (\ref{lin37}), is the principal
generalized eigenvector of the pencil
$$\left(\bS_Q^{\frac{1}{2}}\bg\bg^H\bS_Q^{\frac{1}{2}}+\bI\;,\;
\bS_Q^{\frac{1}{2}}\bh\bh^H\bS_Q^{\frac{1}{2}}+\bI\right) \;,$$
we
obtain~(\ref{lin42}).  The proof is completed by using~(\ref{lin41})
and~(\ref{lin42}) in~(\ref{lin19}).
\end{proof}


\section{Numerical Results}  \label{secV}
\begin{figure}[h]
\centering
\includegraphics[width=3in,height=3in]{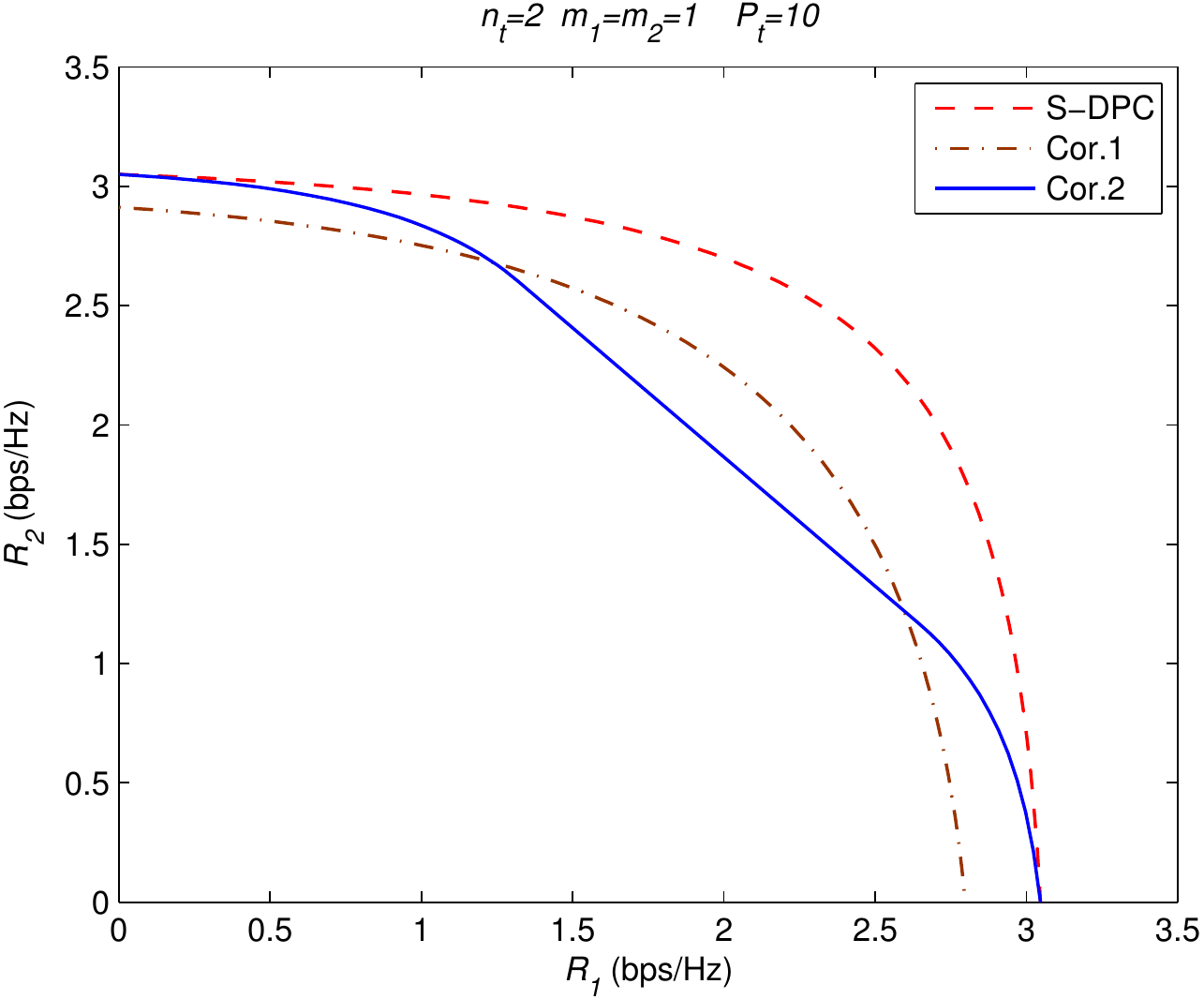}
\caption{Secrecy capacity region of S-DPC together with secrecy rate
region for linear precodings in Cor. 1 and Cor.2, with $P_t=10$, $n_t=2, m_1=m_2=1$.}
\label{lin_exm41}
\end{figure}

In this section, we provide numerical examples to illustrate the
achievable secrecy rate region of the MIMO Gaussian BC under the
average power constraint (\ref{lin2}).  In the first example, we have
$P_t=12$, $\bH=[0.3\; 2.5; 2.2\; 1.8]$ and $\bG=[1.3\; 1.2; 1.5\;
3.9]$, which is identical to the case studied in \cite[Fig. 3
(d)]{LiuLiu10}.  Fig. \ref{lin_exm1} compares the achievable secrecy
rate region of the proposed linear precoding scheme in Section IV-A
with the secrecy capacity region obtained by the optimal S-DPC
approach together with an exhaustive search over suitable matrix
constraints, as described in Section II.  We see that in this example,
the performance of the proposed linear precoding approach is
essentially identical to that of the optimal S-DPC scheme.


In the next example, we study the MISO BC for $P_t=10$.  Fig.~\ref{lin_exm41}
shows the average secrecy rate regions for S-DPC and the suboptimal
linear precoding algorithms described in Corollary~1 and~2.  This plot
is based on an average of over 30000 channel realizations, where the
channel coefficients were generated as independent $\mathcal{CN}(0,1)$
random variables.  We see that Corollary~2 provides near optimal performance
when $\alpha \rightarrow 0$ or $\alpha\rightarrow 1$, while Corollary~1
is better for in-between values of $\alpha$.  The degradation of using
linear precoding with Corollary~1 is never above 15\% for any $\alpha$.

\section{Conclusions}\label{secVI}
We have shown that for a two-user Gaussian BC with an arbitrary number
of antennas at each node, when the channel input is under the matrix power constraint,
linear precoding is optimal and achieves the secrecy
capacity region attained by the optimal S-DPC approach if the matrix
constraint satisfies a specific condition.  We characterized the
form of the linear precoding that achieves the secrecy capacity region in such cases,
and we quantified the maximum loss in secrecy rate that occurs if the matrix
power constraint does not satisfy the given condition.  Based on these
observations, we then formulated a sub-optimal approach for the
general MIMO scenario based on linear precoding for the case of an
average power constraint, for which no known characterization of the
secrecy capacity region exists.  We also studied the MISO case in
detail.  Numerical results indicate that the proposed
linear precoding approaches yield secrecy rate regions that are close to the secrecy capacity achieved by S-DPC.

\appendices

\section{Proof of Eq. (\ref{lin18})}
From (\ref{lin17}), we have
\begin{align}
R_1&=I(V_1;\,\bH(V_1+V_2)+Z_1) -I(V_1;\,\bG V_1+Z_2) \nonumber\\
&=\log\left|\bI+\bS\bH^H\bH\right|-\log\left|\bI+(\bS-\bK^*_t)\bH^H\bH\right|-\log\left|\bI+\bK^*_t\bG^H\bG\right| \;. \label{linap4}
\end{align}
 The covariance $\bK_t^*$, given by (\ref{lin8}), can be rewritten as
\begin{align}\label{linap5}
\bK_t^*&= \bS^{\frac{1}{2}}\left[\bC_1\quad\bC_2\right]
\left[\begin{array}{ccc}(\bC_1^H\bC_1)^{-1} & \b0\\
\b0 & \b0\end{array}\right]\left[\begin{array}{ccc}\bC^H_1\\\bC^H_2\end{array}\right]
\bS^{\frac{1}{2}} \nonumber\\
&=\bS^{\frac{1}{2}}\,\bC_1 (\bC_1^H\bC_1)^{-1} \bC^H_1\, \bS^{\frac{1}{2}} =\bS^{\frac{1}{2}}\,\bP_{\bC_1}\,\bS^{\frac{1}{2}} \; ,
\end{align}
where $\bP_{\bC_1}=\bC_1
(\bC_1^H\bC_1)^{-1} \bC^H_1$ is the
projection matrix onto the column space of $\bC_1$. Moreover, let
$\bP^\perp_{\bC_1}=\bI-\bP_{\bC_1}$ be the
projection onto the space orthogonal to $\bC_1$. Consequently, we have
\begin{align}
\bS-\bK^*_t&=\bS-\bS^{\frac{1}{2}}\,\bP_{\bC_1}\,\bS^{\frac{1}{2}} \nonumber\\
&=\bS^{\frac{1}{2}}\,\bP^\perp_{\bC_1}\,\bS^{\frac{1}{2}} =\bS^{\frac{1}{2}}\,\bP_{\bC_2}\,\bS^{\frac{1}{2}}  \label{linap6}\\
&= \bS^{\frac{1}{2}} \bC\left[\begin{array}{ccc}\b0 & \b0\\
\b0 & (\bC_2^H\bC_2)^{-1}\end{array}\right]\bC^H \bS^{\frac{1}{2}} \; , \label{linap7}
\end{align}
where in (\ref{linap6}), $\bP^\perp_{\bC_1}=\bP_{\bC_2}$ comes from
the fact that $\mathrm{span}\{\bC_1\}\perp\mathrm{span}\{\bC_2\}$, and
$\bC=[\bC_1\quad\bC_2]$ is full-rank.

Following the same steps as in the proof of \cite[Lemma 2]{Weingarten} or \cite[App. B]{LiuLiu10}, we can convert the case when $\bS\succeq\b0$, $|\bS| =0$, to the case where $\bS\succ\b0$ with the same secrecy capacity region.
From (\ref{lin10}) and (\ref{lin11}) we have
\begin{align}\label{linap8}
\begin{split}
&\bH^H\bH= \bS^{-1/2}\left[\bC^{-H}\left[
\begin{array}{ccc}\mathbf{\Lambda}_1 & 0\\
0 & \mathbf{\Lambda}_2\end{array}\right]\bC^{-1}-\bI\right] \bS^{-1/2} \\
&\bG^H\bG= \bS^{-1/2}\left[\bC^{-H}\bC^{-1}-\bI\right] \bS^{-1/2} \; .
\end{split}
\end{align}
Using (\ref{linap7}) and (\ref{linap8}),  we have:
\begin{align}
\left|\bI+(\bS-\bK_t^*)\bH^H\bH\right|&=\left|\bI+\bS^{\frac{1}{2}}\bC \left[
\begin{array}{ccc} \b0 & \b0\\\b0 & (\bC_2^H\bC_2)^{-1}\end{array}\right] \bC^H\cdot\left[\bC^{-H}\left[\begin{array}{ccc}\mathbf{\Lambda}_1 & 0\\
0 & \mathbf{\Lambda}_2\end{array}\right]\bC^{-1}-\bI\right] \bS^{-1/2}\right|\nonumber \\
&=\left|\bI+ \left[\begin{array}{ccc} \b0 & \b0\\\b0 & (\bC_2^H\bC_2)^{-1}\end{array}\right] \cdot\left[\left[\begin{array}{ccc}\mathbf{\Lambda}_1 & 0\\
0 & \mathbf{\Lambda}_2\end{array}\right]-\bC^H\bC\right] \right|   \label{linap9} \\
&=\left|\left[\begin{array}{ccc}\bI & \b0\\\b0 & (\bC_2^H\bC_2)^{-1}\mathbf{\Lambda}_2\end{array}
\right]\right|  \label{linap10} \\
&= \left|(\bC_2^H\bC_2)^{-1}\mathbf{\Lambda}_2\right|
= \left|(\bC_2^H\bC_2)^{-1}\right|\cdot\left|\mathbf{\Lambda}_2\right| \; , \label{linap11}
\end{align}
where (\ref{linap9}) comes from the fact that $\left|\bI+\bA\bB\right|= \left|\bI+\bB\bA\right|$.
Finally, (\ref{linap10}) holds since $\bC_1^H\bC_2=\b0$ and $\bC^H\bC$ is
block diagonal.

Similarly, one can show that
\begin{align} \label{linap13}
\left|\bI+\bK_t^*\bG^H\bG\right|=\left|(\bC_1^H\bC_1)^{-1}\right|
\end{align}
and
\begin{align}
\left|\bI+\bS\bH^H\bH\right|&=\left|\bC^{-H}\left[
\begin{array}{ccc}\mathbf{\Lambda}_1 & 0\\
0 & \mathbf{\Lambda}_2\end{array}\right]\bC^{-1}\right| =\left|(\bC^H\bC)^{-1}\right|\cdot\left|\mathbf{\Lambda}_1\right| \cdot\left|\mathbf{\Lambda}_2\right| \nonumber\\
&=\left|(\bC_1^H\bC_1)^{-1}\right|\cdot\left|(\bC_2^H\bC_2)^{-1}\right|\cdot\left|\mathbf{\Lambda}_1\right| \cdot\left|\mathbf{\Lambda}_2\right| \;.  \label{linap14}
\end{align}

Substituting (\ref{linap11}), (\ref{linap13}) and (\ref{linap14}) in
(\ref{linap4}), we have
$R_1=\log\left|\mathbf{\Lambda}_1\right|=R_1^*$, and
this completes the proof.

\section{Proof of Lemma \ref{lin_lem1}}
From (\ref{lin10})-(\ref{lin11}), we know that $\mathrm{rank}(\bC_1)=b$, where $b$ represents number of generalized eigenvalues of the pencil (\ref{lin9}) that are greater than 1. From (\ref{lin10})-(\ref{lin11}), we have
\begin{align}
&\bC_1^H\left[\bS^{\frac{1}{2}}\bH^H\bH\bS^{\frac{1}{2}}+\bI\right]\bC_1=\mathbf{\Lambda}_1  \label{linap1} \\
&\bC_1^H\left[\bS^{\frac{1}{2}}\bG^H\bG\bS^{\frac{1}{2}}+\bI\right]\bC_1=\bI  \;.\label{linap2}
\end{align}
Subtracting (\ref{linap1}) from (\ref{linap2}), a
straightforward computation yields
\begin{align} \label{linap3}
\bC_1^H\bS^{\frac{1}{2}}\left[\bH^H\bH-\bG^H\bG\right]\bS^{\frac{1}{2}}\bC_1=
\mathbf{\Lambda}_1-\bI  \succ\b0  \;.
\end{align}
From (\ref{linap3}), we have $\bC_1^H\bS^{\frac{1}{2}}\left[\bH^H\bH-\bG^H\bG\right]\bS^{\frac{1}{2}}\bC_1\succ\b0$,
from which it follows that $\mathrm{rank}(\bC_1)=b\le m$. Similarly one can show that  $\mathrm{rank}(\bC_2')\le m'$, where $\bC_2'$ corresponds to the generalized eigenvalues of the pencil (\ref{lin9}) which are less than 1, and $m'$ represents number of negative eigenvalues of $\bH^H\bH-\bG^H\bG$.

\section{Proof of Theorem \ref{lin_thm2}}
We want to characterize the matrices $\bS\succeq\b0$ for which
\begin{align} \label{linap22}
\left(\bS^{\frac{1}{2}}\bH^H\bH\bS^{\frac{1}{2}}+\bI\;,\; \bS^{\frac{1}{2}}\bG^H\bG\bS^{\frac{1}{2}}+\bI\right) 
\end{align}
has generalized eigenvectors with orthogonal $\bC_1$ and $\bC_2$. For
any positive semidefinite matrix $\bS\in \mathbb{C}^{n_t\times n_t}$,
there exists a matrix $\bT\in \mathbb{C}^{n_t \times n_t}$ such that
$\bS=\bT\bT^H$ \cite{Horn}.  More precisely,
$\bT=\bS^{\frac{1}{2}}\bPsi$, where $\bPsi$ can be any $n_t\times n_t$
unitary matrix; thus $\bT$ is not unique.

\begin{rem}\label{lin_remap2}
Let the invertible matrix $\overline{\bC}$ and the diagonal matrix $\overline{\mathbf{\Lambda}}$ respectively represent the generalized eigenvectors and eigenvalues of
\begin{align} \label{linap23}
\left(\bT^H\bH^H\bH\bT+\bI\;,\; \bT^H\bG^H\bG\bT+\bI\right) \; ,
\end{align}
so that
\begin{align}\label{linap_ex1}
\begin{split}
&\overline{\bC}^H\left[\bT^H\bH^H\bH\bT+\bI\right]\overline{\bC}=\overline{\mathbf{\Lambda}}\\
&\overline{\bC}^H\left[\bT^H\bG^H\bG\bT+\bI\right]\overline{\bC}=\bI \; ,
\end{split}
\end{align}
where $\bT=\bS^{\frac{1}{2}}\bPsi$ for a given unitary matrix
$\bPsi$. By comparing (\ref{lin10}) and (\ref{linap_ex1}), one can
confirm that $\bPsi\overline{\bC} =\bC$ and
$\overline{\mathbf{\Lambda}} = \mathbf{\Lambda}$, where $\bC$ and
$\mathbf{\Lambda}$ are respectively the generalized eigenvectors
and eigenvalues of~(\ref{linap22}), as
given by (\ref{lin10}).
\end{rem}

Also note that, for any unitary $\bPsi$,
$\overline{\bC}^H\overline{\bC} =\bC^H\bC$. Thus, finding a
$\bS\succeq\b0$ such that~(\ref{linap22}) has orthogonal
$\bC_1$ and $\bC_2$ (block diagonal $\bC^H\bC$) is equivalent to
finding a $\bT$, $\bS=\bT\bT^H$, such that~(\ref{linap23}) has
orthogonal $\overline{\bC}_1$ and $\overline{\bC}_2$ (block diagonal
$\overline{\bC}^H\overline{\bC}$).

The \emph{if} part of Theorem \ref{lin_thm2} is easy to show. We want to show that if
$\bS=\bT\bT^H$ and $\bT$ simultaneously block diagonalizes $\bH^H\bH$
and $\bG^H\bG$, as given by (\ref{linshrink1})
such that $\bK_{\bH1}\succeq\bK_{\bG1}$ and $\bK_{\bH2}\preceq\bK_{\bG2}$,
then $\bC_1^H\bC_2=\b0$. From the definition of the generalized eigenvalue decomposition, we have
\begin{align}\label{linap24}
\begin{split}
&\overline{\bC}^H\left[\bT^H\bH^H\bH\bT+\bI\right]\overline{\bC}= \overline{\bC}^H \left[\begin{array}{ccc}\bI+\bK_{\bH1} & \b0\\
\b0 & \bI+\bK_{\bH2}\end{array}\right]\overline{\bC}= \left[\begin{array}{ccc}\bD_1 & \b0\\
\b0 & \bD_2 \end{array}\right]\\
&\overline{\bC}^H\left[\bT^H\bG^H\bG\bT+\bI\right]\overline{\bC}=\overline{\bC}^H \left[\begin{array}{ccc}\bI+\bK_{\bG1} & \b0\\
\b0 & \bI+\bK_{\bG2}\end{array}\right]\overline{\bC}= \left[\begin{array}{ccc}\bI & \b0\\
\b0 & \bI \end{array}\right] \; , \\
\end{split}
\end{align}
from which we have
$$\overline{\bC}=
\left[\overline{\bC}_1\quad\overline{\bC}_2\right]=\left[\begin{array}{ccc}\overline{\bC}_{11}
& \b0\\ \b0 & \overline{\bC}_{22}\end{array}\right] \;,$$
where the invertible matrix $\overline{\bC}_{11}$ and diagonal matrix
$\bD_1$ are respectively the generalized eigenvectors and eigenvalues of
$\left(\bI+\bK_{\bH1}\,;\,\bI+\bK_{\bG1}\right) $. Since $\bK_{\bH1}\succeq \bK_{\bG1}$,
then $\bD_1\succeq\bI$, which shows that $\overline{\bC}_1$ corresponds to generalized eigenvalues that are bigger than or equal to one. We have a similar
definition for $\overline{\bC}_{22}$ and diagonal matrix $\bD_2$,
corresponding to
$\left(\bI+\bK_{\bH2}\,;\,\bI+\bK_{\bG2}\right) $.
Finally, since $\overline{\bC}^H\overline{\bC}$ is
block diagonal, then $\bC^H\bC$, where $\bC$ is the generalized
eigenvector matrix of (\ref{linap22}), is block diagonal as well. This
completes the \emph{if} part of the theorem.

In the following, we prove the \emph{only if} part of Theorem \ref{lin_thm2};
{\em i.e.,} we show that if $\bS\succeq\b0$ results in (\ref{linap22}) having
orthogonal $\bC_1$ and $\bC_2$, then there must exist a square matrix
$\bT$ such that $\bS=\bT\bT^H$ and $\bT^H\bH^H\bH\bT $ and
$\bT^H\bG^H\bG\bT$ are simultaneously block diagonalized as in~(\ref{linshrink1})
with $\bK_{\bH1}\succeq\bK_{\bG1}$ and $\bK_{\bH2}\preceq\bK_{\bG2}$.

Let $\bS^{\frac{1}{2}}\bG^H\bG\bS^{\frac{1}{2}}$ have the eigenvalue
decomposition $\bPhi_B\bSig_B\bPhi_B^H$, where $\bPhi_B$ is unitary
and $\bSig_B$ is a positive semidefinite diagonal matrix. Also let
$(\bI+\bSig_B)^{-\frac{1}{2}}\bPhi_B^H\left(\bI+\bS^{\frac{1}{2}}\bH^H\bH\bS^{\frac{1}{2}}\right)\bPhi_B(\bI+\bSig_B)^{-\frac{1}{2}}$
have the eigenvalue decomposition $\bPhi_A\bSig_A\bPhi_A^H$, where
$\bPhi_A$ is unitary and $\bSig_A$ is a positive definite diagonal
matrix. One can easily confirm that \cite{Horn}
$\bC=\bPhi_B(\bI+\bSig_B)^{-\frac{1}{2}}\bPhi_A$
and $\mathbf{\Lambda}=\bSig_A$, where $\bC$ and $\mathbf{\Lambda}$ are
respectively the generalized eigenvectors and eigenvalues of~(\ref{linap22}).  Also let $\bC$ be ordered such that $\bC=[\bC_1\quad\bC_2]$, where $\bC_1$ corresponds to the generalized eigenvalues bigger than
(or equal to) 1. We have
\begin{align}\label{linap25}
\bC^H\bC= \bPhi_A ^H(\bI+\bSig_B)^{-1}\bPhi_A  \;.
\end{align}

From (\ref{linap25}), $\bC^H\bC$ is block diagonal iff the unitary matrix
$\bPhi_A$ is block diagonal. Recalling that $\bPhi_A$ is the
eigenvector matrix of
$(\bI+\bSig_B)^{-\frac{1}{2}}\bPhi_B^H\left(\bI+\bS^{\frac{1}{2}}\bH^H\bH\bS^{\frac{1}{2}}\right)\bPhi_B(\bI+\bSig_B)^{-\frac{1}{2}}$,
a block diagonal $\bPhi_A$ leads to
$(\bI+\bSig_B)^{-\frac{1}{2}}\bPhi_B^H\left(\bI+\bS^{\frac{1}{2}}\bH^H\bH\bS^{\frac{1}{2}}\right)\bPhi_B(\bI+\bSig_B)^{-\frac{1}{2}}$,
and consequently
$\bPhi_B^H\bS^{\frac{1}{2}}\bH^H\bH\bS^{\frac{1}{2}}\bPhi_B$ must be
block diagonal. Thus, if $\bC^H\bC$ is block diagonal,
{\em i.e.,} $\bC_1^H\bC_2=\b0$, there must exist a unitary matrix $\bPhi_B$
such that $\bPhi_B^H\bS^{\frac{1}{2}}\bH^H\bH\bS^{\frac{1}{2}}\bPhi_B$
and $\bPhi_B^H\bS^{\frac{1}{2}}\bG^H\bG\bS^{\frac{1}{2}}\bPhi_B$ are
simultaneously block diagonal.\footnote{Note that
$\bPhi_B^H\bS^{\frac{1}{2}}\bG^H\bG\bS^{\frac{1}{2}}\bPhi_B$ is
actually diagonal, and hence also block diagonal.} Letting $\bT=\bS^{\frac{1}{2}}\bPhi_B$ results in~(\ref{linap24}), for which we must have $\bK_{\bH1}\succeq\bK_{\bG1}$ and $\bK_{\bH2}\preceq\bK_{\bG2}$, otherwise it contradicts the ordering of $\bC=[\bC_1\quad\bC_2]$. This completes the proof.

\section{Proof of Theorem \ref{lin_lem2}}
We need to prove that the secrecy rate pair $(R_1,R_2)$ given by (\ref{lin19}) is achievable.

\begin{rem}\label{lin_remap1} By applying the Schur Complement Lemma \cite{Horn}
on
$$\bC^H\bC=\left[\bC_1 \quad \bC_2\right]^H\left[\bC_1\quad \bC_2\right]=\left[
\begin{array}{ccc}\bC_1^H\bC_1 & \bC_1^H\bC_2\\
\bC_2^H\bC_1 & \bC_2^H\bC_2\end{array}\right]$$
and recalling the fact that $\bC$ is full-rank, we have that
$\bC_2^H\bC_2-\bC_2^H\bC_1(\bC_1^H\bC_1)^{-1}\bC_1=\bC_2^H\bP^\perp_{\bC_1}\bC_2$
is full rank. Similarly, one can show that
$\left(\bC_1^H\bP^\perp_{\bC_2}\bC_1\right)^{-1}$ exists. Also, we have
$\left|\bC^H\bC\right|=
\left|\bC_1^H\bP^\perp_{\bC_2}\bC_1\right|\cdot\left|\bC_2^H\bC_2\right|$
$=\left|\bC_2^H\bP^\perp_{\bC_1}\bC_2\right|\cdot\left|\bC_1^H\bC_1\right|$.
\end{rem}

Define $\widehat{\bC}=[\bP^\perp_{\bC_2}\bC_1 \quad \bC_2]$, so that
$$\widehat{\bC}^H\widehat{\bC}=\left[\bP^\perp_{\bC_2}\bC_1 \quad \bC_2\right]^H\left[\bP^\perp_{\bC_2}\bC_1\quad \bC_2\right]=\left[
\begin{array}{ccc}\bC_1^H\bP^\perp_{\bC_2}\bC_1 & \b0\\
\b0 & \bC_2^H\bC_2\end{array}\right]\;.$$
Consequently, we can write
\begin{align}\label{linap15}
\bP_{\bC_2}&= \bC_2(\bC_2^H\bC_2)^{-1}\bC_2^H
=\widehat{\bC}\left[\begin{array}{ccc}\b0 & \b0\\
\b0 & (\bC_2^H\bC_2)^{-1}\end{array}\right]\widehat{\bC}^H
\end{align}
and
\begin{align}\label{linap16}
\bP^\perp_{\bC_2}&= \bI-\bP_{\bC_2} = \widehat{\bC} \left(\widehat{\bC}^H\widehat{\bC}\right)^{-1}\widehat{\bC}^H -\bP_{\bC_2}  \nonumber\\
&=\widehat{\bC}\left[\begin{array}{ccc}\left(\bC_1^H\bP^\perp_{\bC_2}\bC_1\right)^{-1} & \b0\\
\b0 & \b0\end{array}\right]\widehat{\bC}^H \; .
\end{align}

In the following we show the achievablity of $R_1$ in  (\ref{lin19}). The achievablity of $R_2$ is obtained in a similar manner.
Since $V_1$ and $V_2$ in Theorem~\ref{lin_lem2} are independent, from (\ref{lin17}) we have
\begin{align}\label{linap17}
R_1&=I(V_1;\,\bH(V_1+V_2)+Z_1) -I(V_1;\,\bG V_1+Z_2) \nonumber\\
&=\log\left|\bI+\bH\bS\bH^H\right|-\log\left|\bI+\bH(\bS^{\frac{1}{2}}\bP_{\bC_2}\bS^{\frac{1}{2}})\bH^H\right|-\log\left|\bI+\bG(\bS^{\frac{1}{2}}\bP^\perp_{\bC_2}\bS^{\frac{1}{2}})\bG^H\right| \nonumber\\
&=\log\left|\bI+\bS\bH^H\bH\right|-\log\left|\bI+\bS^{\frac{1}{2}}\bP_{\bC_2}\bS^{\frac{1}{2}}\bH^H\bH\right|-\log\left|\bI+\bS^{\frac{1}{2}}\bP^\perp_{\bC_2}\bS^{\frac{1}{2}}\bG^H\bG\right|  \;.
\end{align}
Recalling (\ref{linap8}), we have
\begin{align}
\left|\bI+\bS\bH^H\bH\right|&=\left|(\bC^H\bC)^{-1}\right|\cdot\left|\mathbf{\Lambda}_1\right| \cdot\left|\mathbf{\Lambda}_2\right| \nonumber\\
&=\left|(\bC_1^H\bP^\perp_{\bC_2}\bC_1)^{-1}\right|\cdot\left|(\bC_2^H\bC_2)^{-1}\right|\cdot\left|\mathbf{\Lambda}_1\right| \cdot\left|\mathbf{\Lambda}_2\right| \; , \label{linap18}
\end{align}
where we used Remark \ref{lin_remap1} to obtain (\ref{linap18}). From (\ref{linap15}), we have
\begin{align}
&\log\left|\bI+\bS^{\frac{1}{2}}\bP_{\bC_2}\bS^{\frac{1}{2}}\bH^H\bH\right|\nonumber\\
&=\left|\bI+\widehat{\bC} \left[
\begin{array}{ccc} \b0 & \b0\\\b0 & (\bC_2^H\bC_2)^{-1}\end{array}\right] \widehat{\bC}^H\cdot\left[\bC^{-H}\left[\begin{array}{ccc}\mathbf{\Lambda}_1 & 0\\
0 & \mathbf{\Lambda}_2\end{array}\right]\bC^{-1}-\bI\right] \right|\nonumber \\
&=\left|\bI+\left[\begin{array}{ccc} \b0 & \b0\\\b0 & (\bC_2^H\bC_2)^{-1}\end{array}\right] \cdot\left[\widehat{\bC}^H\bC^{-H}\left[\begin{array}{ccc}\mathbf{\Lambda}_1 & 0\\
0 & \mathbf{\Lambda}_2\end{array}\right]\bC^{-1}\widehat{\bC}-\widehat{\bC}^H\widehat{\bC} \right] \right|\nonumber \\
&=\left|\bI+\left[\begin{array}{ccc} \b0 & \b0\\\b0 & (\bC_2^H\bC_2)^{-1}\end{array}\right] \cdot\left[\left[\begin{array}{ccc}\bI & \bN^H\\
\b0 & \bI\end{array}\right] \left[\begin{array}{ccc}\mathbf{\Lambda}_1 & 0\\
0 & \mathbf{\Lambda}_2\end{array}\right]\left[\begin{array}{ccc}\bI & \b0\\
\bN & \bI\end{array}\right] -\widehat{\bC}^H\widehat{\bC} \right] \right|       \label{linap19}\\
&=\left|\bI+\left[\begin{array}{ccc} \b0 & \b0\\\b0 & (\bC_2^H\bC_2)^{-1}\end{array}\right] \cdot\left[\left[\begin{array}{ccc}\mathbf{\Lambda}_1+\bN^H\mathbf{\Lambda}_2\bN & \bN^H\mathbf{\Lambda}_2\\
\mathbf{\Lambda}_2\bN & \mathbf{\Lambda}_2\end{array}\right] -\widehat{\bC}^H\widehat{\bC} \right] \right|  \nonumber\\
&=\left|\left[\begin{array}{ccc} \bI & \b0 \\  (\bC_2^H\bC_2)^{-1}\mathbf{\Lambda}_2\bN & (\bC_2^H\bC_2)^{-1}\mathbf{\Lambda}_2\end{array}\right]  \right| \nonumber\\
&= \left|(\bC_2^H\bC_2)^{-1}\mathbf{\Lambda}_2\right| =\left|(\bC_2^H\bC_2)^{-1}\right|\cdot\left|\mathbf{\Lambda}_2\right| \; ,   \label{linap20}
\end{align}
where in (\ref{linap19}), $\bN=\left(\bC_2^H\bP^\perp_{\bC_1}\bC_2\right)^{-1}\bC_2^H\bP^\perp_{\bC_1}\bP^\perp_{\bC_2}\bC_1$, and we used the fact that
\begin{align*}
\bC^{-1}\widehat{\bC}&=\left[\begin{array}{ccc}(\bC_1^H\bP^\perp_{\bC_2}\bC_1)^{-1} \bC_1^H\bP^\perp_{\bC_2}\\
(\bC_2^H\bP^\perp_{\bC_1}\bC_2)^{-1} \bC_2^H\bP^\perp_{\bC_1} \end{array}\right]\, \left[\bP^\perp_{\bC_2}\bC_1 \quad \bC_2\right]
= \left[\begin{array}{ccc}\bI & \b0\\ \bN & \bI\end{array}\right]  \;.
\end{align*}

Similarly, we have
\begin{align}
&\log\left|\bI+\bS^{\frac{1}{2}}\bP^\perp_{\bC_2}\bS^{\frac{1}{2}}\bG^H\bG\right|\nonumber\\
&=\left|\bI+\left[\begin{array}{ccc} (\bC_1^H\bP^\perp_{\bC_2}\bC_1)^{-1} & \b0\\\b0 & \b0\end{array}\right] \cdot\left[\left[\begin{array}{ccc}\bI & \bN^H\\
\b0 & \bI\end{array}\right] \left[\begin{array}{ccc}\bI & \b0\\
\bN & \bI\end{array}\right] -\widehat{\bC}^H\widehat{\bC} \right] \right|  \nonumber\\
&= \left|\bI+\left[\begin{array}{ccc} (\bC_1^H\bP^\perp_{\bC_2}\bC_1)^{-1} & \b0\\\b0 & \b0\end{array}\right] \cdot\left[\left[\begin{array}{ccc}\bI+\bN^H\bN & \bN^H\\ \bN & \bI\end{array}\right] -\widehat{\bC}^H\widehat{\bC} \right] \right|  \nonumber\\
&= \left|\left[\begin{array}{ccc} (\bC_1^H\bP^\perp_{\bC_2}\bC_1)^{-1} (\bI+\bN^H\bN) &  (\bC_1^H\bP^\perp_{\bC_2}\bC_1)^{-1}\bN^H \\ \b0 & \bI\end{array}\right] \right|  \nonumber\\
&= \left|(\bC_1^H\bP^\perp_{\bC_2}\bC_1)^{-1}\right|\cdot\left|\bI+\bN^H\bN\right|    \;. \label{linap21}
\end{align}
Subsituting (\ref{linap18}), (\ref{linap20}) and (\ref{linap21}) in (\ref{linap17}), we have
$R_1= \max(0, \,\log\left|\mathbf{\Lambda}_1\right|-\log\left|\bI+\bN^H\bN\right|)$, which completes the proof.

\section{Proof of Lemma \ref{lin_lem3}}
We want to show that $\bW\bH^H\bH\bW$ and $\bW\bG^H\bG\bW$ commute, where $\bW=(\bH^H\bH+\bG^H\bG)^{-\frac{1}{2}}$
Let the invertible matrix $\widehat{\bC}$ and diagonal matrix
$\widehat{\mathbf{\Lambda}}$ respectively represent the generalized
eigenvectors and eigenvalues of
$\left(\bW\bH^H\bH\bW+\bI\;;\; \bW\bG^H\bG\bW+\bI\right)$, so that
\begin{align}
&\widehat{\bC}^H\left[\bW\bH^H\bH\bW+\bI\right]\widehat{\bC}=\widehat{\mathbf{\Lambda}}  \label{linap26}\\
&\widehat{\bC}^H\left[\bW\bG^H\bG\bW+\bI\right]\widehat{\bC}=\bI  \;. \label{linap27}
\end{align}
Adding (\ref{linap26}) and (\ref{linap27}), we have
\begin{align*}
\widehat{\bC}^H\left[\bW(\bH^H\bH+\bG^H\bG)\bW+2\,\bI\right]\widehat{\bC}=
3\,\widehat{\bC}^H\widehat{\bC}=(\widehat{\mathbf{\Lambda}}+\bI) \;,
\end{align*}
from which it results that $\widehat{\bC}$ must be of the form \cite{Horn}
\begin{align}\label{linap28}
\widehat{\bC}=\frac{1}{\sqrt{3}}\,\bPhi_{\bw}(\widehat{\mathbf{\Lambda}}+\bI)^{\frac{1}{2}} \;,
\end{align}
where $\bPhi_{\bw}$ is an unknown unitary matrix. In the following, as we continue the proof, $\bPhi_{\bw}$ is characterized too.

Substituting (\ref{linap28}) in (\ref{linap26}) and (\ref{linap27}),
it is revealed that the unitary matrix $\bPhi_{\bw}$ represents the
common set of eigenvectors for the matrices $\bW\bH^H\bH\bW+\bI$ and
$\bW\bG^H\bG\bW+\bI$, and thus both matrices commute.  In particular,
\begin{align*}
\begin{split}
\bPhi_{\bw}^H\left[\bW\bH^H\bH\bW+\bI\right]\bPhi_{\bw} &=3\,\widehat{\mathbf{\Lambda}} (\widehat{\mathbf{\Lambda}}+\bI)^{-1}=3\,(\widehat{\mathbf{\Lambda}}^{-1}+\bI)^{-1} \\
\bPhi_{\bw}^H\left[\bW\bG^H\bG\bW+\bI\right]\bPhi_{\bw}&=3\, (\widehat{\mathbf{\Lambda}}+\bI)^{-1}\;.
\end{split}
\end{align*}
Consequently, $\bSig_1$ and $\bSig_2$ are diagonal:
\begin{align}\label{linap29}
\begin{split}
\bPhi_{\bw}^H\;\bW\bH^H\bH\bW\;\bPhi_{\bw} &=3\,(\widehat{\mathbf{\Lambda}}^{-1}+\bI)^{-1} -\bI =\bSig_1\\
\bPhi_{\bw}^H\,\bW\bG^H\bG\bW\,\bPhi_{\bw}&=3\, (\widehat{\mathbf{\Lambda}}+\bI)^{-1}-\bI=\bSig_2\;.
\end{split}
\end{align}
It is interesting to note that, since $\bPhi_{\bw}^H\,\bW\bH^H\bH\bW\,\bPhi_{\bw}\succeq\b0$ and $\bPhi_{\bw}^H\,\bW\bG^H\bG\bW\,\bPhi_{\bw}\succeq\b0$, we have
$\frac{1}{2}\bI\preceq\widehat{\mathbf{\Lambda}}\preceq2\,\bI$.

\section{Proof of Lemma \ref{lin_lem4}}
We first consider the generalized eigenvalue decomposition for
\begin{align} \label{linap31}
\left(\bT^H_{\bw}\bH^H\bH\bT_{\bw}+\bI\;,\; \bT^H_{\bw}\bG^H\bG\bT_{\bw}+\bI\right)  \;,
\end{align}
where $\bT_{\bw}$ is given by (\ref{lin23}) and
\begin{align*}
\begin{split}
&\overline{\bC}_{\bw}^H\left[\bT_{\bw}^H\bH^H\bH\bT_{\bw}+\bI\right]\overline{\bC}_{\bw}=\mathbf{\Lambda}_{\bw}\\
&\overline{\bC}_{\bw}^H\left[\bT_{\bw}^H\bG^H\bG\bT_{\bw}+\bI\right]\overline{\bC}_{\bw}=\bI \; .
\end{split}
\end{align*}
Using (\ref{lin21}), and noting that $\bPhi_{\bw}$ is unitary and $\bP$ is diagonal, a straightforward calculation yields
\begin{align*}
\begin{split}
&\overline{\bC}_{\bw}^H\left[\bSig_1\bP+\bI\right]\overline{\bC}_{\bw}=\mathbf{\Lambda}_{\bw}\\
&\overline{\bC}_{\bw}^H\left[\bSig_2\bP+\bI\right]\overline{\bC}_{\bw}=\bI\;,
\end{split}
\end{align*}
where $\bSig_1$ and $\bSig_2$ are respectively (diagonal) eigenvalue matrices of $\bW\bH^H\bH\bW$ and $\bW\bG^H\bG\bW$, as given by (\ref{lin21}).
Thus, $\overline{\bC}_{\bw}$ is diagonal and is given by
\begin{align}\label{linap32}
\overline{\bC}_{\bw}&=\left(\bSig_2\bP+\bI\right)^{-\frac{1}{2}}  \;.
\end{align}
Consequently, we have
$\mathbf{\Lambda}_{\bw}= \left(\bSig_2\bP+\bI\right)^{-1}   \left(\bSig_1\bP+\bI\right)$.

Let $\sigma_{1i}$, $\sigma_{2i}$ and $p_i$ represent the $i^{\rm th}$ diagonal elements of $\bSig_1$, $\bSig_2$ and $\bP$, respectively.   We note that for any $p_i$, $(1+\sigma_{1i}\,p_i)/(1+\sigma_{2i}\,p_i)>1$ iff $\sigma_{1i}>\sigma_{2i}$. Thus, based on the argument that we made after Lemma \ref{lin_lem3}, the first $\rho$ diagonal elements of $\mathbf{\Lambda}_{\bw}$ represent generalized eigenvalues greater than 1. Letting
\begin{align} \label{linap34}
\bP=\left[\begin{array}{ccc}\bP_1 & \b0\\\b0 & \bP_2\end{array}\right]
\end{align}
where $\bP_1$ is $\rho\times\rho$ and $\bP_2$ is $(n_t-\rho)\times(n_t-\rho)$, we have:
\begin{align} \label{linap35}
\mathbf{\Lambda}_{\bw}&= \left[\begin{array}{ccc}\left(\bSig_{2\rho}\bP_1+\bI\right)^{-1}   \left(\bSig_{1\rho}\bP_1+\bI\right) & \b0\\\b0 & \left(\bSig_{2\bar{\rho}}\bP_2+\bI\right)^{-1}   \left(\bSig_{1\bar{\rho}}\bP_2+\bI\right)\end{array}
\right]  = \left[\begin{array}{ccc}\mathbf{\Lambda}_{1\bw} & \b0\\\b0 & \mathbf{\Lambda}_{2\bw}\end{array}
\right] \; ,
\end{align}
where $\bSig_{i\rho}$ and $\bSig_{i\bar{\rho}}$ ($i=1,2$) are given by (\ref{lin22}).  Consequently, (\ref{linap32}) can be rewritten as
\begin{align} \label{linap36}
\overline{\bC}_{\bw}=[\overline{\bC}_{1\bw}\quad\overline{\bC}_{2\bw}]= \left[\begin{array}{ccc}\left(\bSig_{2\rho}\bP_1+\bI\right)^{-\frac{1}{2}} & \b0\\\b0 & \left(\bSig_{2\bar{\rho}}\bP_2+\bI\right)^{-\frac{1}{2}}\end{array}\right]  \;.
\end{align}

From the argument before Lemma \ref{lin_lem4}, for any diagonal

$\bP\succeq\b0$, linear precoding is an optimal solution for the BC under the
matrix power constraint $\bS_{\bw}=\bT_{\bw}\bT_{\bw}^H$, where
$\bT_{\bw}$ is given by (\ref{lin23}). More precisely, from Theorem
\ref{lin_thm1}, $X=V_1+V_2$ is optimal, where $V_1$ and $V_2$ are
independent Gaussian precoders, respectively corresponding to $W_1$
and $W_2$ with zero means and covariance matrices $\bK^*_{t\bw}$ and
$\bS_{\bw}-\bK^*_{t\bw}$, where $\bK^*_{t\bw}$ is given by
\begin{align} \label{linap37}
\bK^*_{t\bw}= \bS_{\bw}^{\frac{1}{2}}\bC_{\bw} \left[\begin{array}{ccc}(\bC_{1\bw}^H\bC_{1\bw})^{-1} & \b0\\\b0 & \b0\end{array}
\right]\bC_{\bw}^H \bS_{\bw}^{\frac{1}{2}}
\end{align}
and
$\bC_{\bw}$ is the generalized eigenvector matrix for
\begin{align} \label{linap38}
\left(\bS^{\frac{1}{2}}_{\bw}\bH^H\bH\bS^{\frac{1}{2}}_{\bw}+\bI\;,\; \bS^{\frac{1}{2}}_{\bw}\bG^H\bG\bS^{\frac{1}{2}}_{\bw}+\bI\right)   \;.
\end{align}

We note that there exists a unitary matrix $\bPsi$ for which $\bS^{\frac{1}{2}}_{\bw} = \bT_{\bw}\bPsi^H$ \cite{Horn}, where $\bS_{\bw}=\bT_{\bw}\bT^H_{\bw}$. We also note that, from Remark \ref{lin_remap2}, $\bC_{\bw} =\bPsi\overline{\bC}_{\bw}$ and $\bC^H_{\bw}\bC_{\bw} =\overline{\bC}^H_{\bw}\overline{\bC}_{\bw}$. Thus, $\bK^*_{t\bw}$ can be rewritten as
\begin{align}
\bK^*_{t\bw}&= \bT_{\bw}\overline{\bC}_{\bw}\left[\begin{array}{ccc}(\overline{\bC}_{1\bw}^H\overline{\bC}_{1\bw})^{-1} & \b0\\\b0 & \b0\end{array}
\right]\overline{\bC}_{\bw}^H \bT_{\bw}^H
= \bT_{\bw}\left[\begin{array}{ccc}\bI & \b0\\\b0 & \b0\end{array}
\right]\bT_{\bw}^H  \label{linap39} \\
&= \bW\bPhi_{\bw}\bP^{\frac{1}{2}}\left[\begin{array}{ccc}\bI & \b0\\\b0 & \b0\end{array}
\right]\bP^{\frac{1}{2}}\bPhi^H_{\bw}\bW
= \bW\bPhi_{\bw}\left[\begin{array}{ccc}\bP_1 & \b0\\\b0 & \b0\end{array}
\right]\bPhi^H_{\bw}\bW  \; , \label{linap40}
\end{align}
where (\ref{linap39}) comes from (\ref{linap36}), and (\ref{linap40}) comes from (\ref{linap34}).
Consequently, $\bS_{\bw}-\bK^*_{t\bw}$ can be written as
\begin{align}
\bS_{\bw}-\bK^*_{t\bw}= \bT_{\bw}\bT^H_{\bw}-\bK^*_{t\bw}
&=\bW\bPhi_{\bw} \bP\bPhi^H_{\bw}\bW-\bK^*_{t\bw}\nonumber\\
&= \bW\bPhi_{\bw}\left[\begin{array}{ccc}\b0 & \b0\\\b0 & \bP_2\end{array}
\right]\bPhi^H_{\bw}\bW  \; . \label{linap41}
\end{align}

From (\ref{linap40}) and (\ref{linap41}), under the
matrix power constraint $\bS_{\bw}$ given by (\ref{lin24}), the optimal
linear precoding is $X=V_1+V_2$, where precoding signals $V_1$ and
$V_2$ are independent Gaussian vectors with zero means and covariance
matrices given by (\ref{linap40}) and (\ref{linap41}),
respectively. Alternatively, the optimal precoder can be
represented as $X=\bW\bPhi_{\bw}
\left[\begin{array}{c} V'_1 \\ V'_2\end{array}\right]$, where
precoding signals $V'_1$ and $V'_2$ are independent Gaussian vectors
with zero means and diagonal covariance matrices respectively given by
$\bP_1$ and $\bP_2$. In both cases $\mathbb{E}\{XX^H\}=\bS_{\bw}$, and
the same secrecy rate region is achieved.

\begin{singlespace}

\end{singlespace}
\end{document}